\documentclass[11pt]{article}
\usepackage{subcaption}

\usepackage[preprint]{jmlr2e-arxiv}
\usepackage[utf8]{inputenc}
\usepackage{dsfont}
\setcitestyle{numbers,square,comma,sort&compress}

\PassOptionsToPackage{numbers,compress}{natbib}
\usepackage{natbib}
\usepackage{graphicx,graphics}
\usepackage{wrapfig}
\usepackage{amsmath}
\usepackage{enumitem}
\usepackage{booktabs} 
\usepackage{nicefrac} 
\usepackage{yfonts}
\usepackage[euler]{textgreek}
\usepackage{listings}
\usepackage{MnSymbol}
\usepackage{silence}
\usepackage[mathscr]{eucal}
\usepackage[linesnumbered, ruled, vlined]{algorithm2e}%
\usepackage{algpseudocode}

\def\P{\mathbb{P}}

\def\R{\mathbb{R}}

\def\P{\mathbb{P}}

\newcommand{\argmax}{\mathop{\mathrm{argmax}}}

\newcommand{\red}[1]{\textcolor{red}{#1}}
\newcommand{\teal}[1]{\textcolor{teal}{#1}}
\newcommand{\flink}[1]{\textcolor{blue}{\underline{#1}}}
\usepackage{listings}
\usepackage{xcolor}

\newcounter{counter}[section]

\usepackage{tcolorbox}
\usepackage[T1]{fontenc}
\usepackage{hyperref}
\usepackage{url}
\usepackage{amsfonts}
\usepackage{microtype}
\usepackage{xspace}

\begin{document}

\title{Mechanism Design for Quality-Preserving  LLM Advertising}

\author{\name Jiale Han \email jialehan@ucla.edu\\
\addr Department of Statistics and Data Science\\
University of California, Los Angeles, CA 90095-1554, USA\\
\AND
\name Xiaowu Dai\thanks{\textit{Address for correspondence:} Xiaowu Dai, Department of Statistics, University of California, Los Angeles, 8125 Math Sciences Bldg \#951554, Los Angeles, CA 90095, USA. Email: daix@ucla.edu.}
\email daix@ucla.edu\\
\addr Department of Statistics and Data Science and Department of Biostatistics\\
University of California, Los Angeles, CA 90095-1554, USA
}

\maketitle

\begin{abstract}
\noindent
Embedding advertisements into large language model (LLM) outputs introduces a fundamental tension: revenue optimization can distort content and degrade user experience. Existing approaches largely ignore this trade-off, often forcing irrelevant ads into responses. We propose a \emph{quality-preserving} auction framework that explicitly integrates content fidelity into the mechanism design. Built on retrieval-augmented generation (RAG), our approach treats \emph{organic content} as a reference and derives an endogenous reserve price that screens out ads with non-positive marginal social welfare contributions. We develop a KL-regularized single-allocation mechanism with Myerson payments and a screened VCG multi-allocation mechanism, both satisfying dominant-strategy incentive compatibility and individual rationality. Experiments across diverse scenarios demonstrate that our mechanisms outperform existing baselines in metrics such as revenue per ad and semantic similarity to no-ad responses. Our results establish a new paradigm for LLM advertising that enables monetization without compromising output quality.

\end{abstract}
\bigskip

\section{Introduction}\label{sec:intro}
Large language models (LLMs) are rapidly becoming the primary interface for information access, enabling applications such as question answering, content generation, and decision support~\citep{achiam2023gpt, brown2020language,  touvron2023llama, zhao2023survey}. Their training and deployment, however, incur substantial computational cost, creating an urgent need for scalable monetization~\citep{minaee2024large}. Online advertising is a natural solution~\citep{choi2020online}: platforms can allow advertisers to bid for inclusion in LLM-generated responses. Unlike traditional display advertising, this paradigm embeds ads directly into generated content, introducing a fundamental \emph{tension} between revenue maximization and output quality. Recent work proposes auction mechanisms for LLM-based advertising~\citep{dubey2024auctions, duetting2024mechanism, hajiaghayi2024ad, soumalias2024truthful}, but largely overlooks this tension. Existing approaches optimize welfare based only on advertiser value and restrict generation to sponsored content, which can introduce irrelevant ads and degrade coherence~\citep{duetting2024mechanism, hajiaghayi2024ad}. While \citet{soumalias2024truthful} address this via distributional penalties, their method incurs inference costs that scale with the number of advertisers.
A~key challenge remains: \emph{how to design auction mechanisms that enable monetization without compromising the semantic fidelity of LLM outputs?}

In this paper, we propose a \emph{quality-preserving} mechanism design framework for LLM advertising. Our approach builds on retrieval-augmented generation (RAG)~\citep{lewis2020retrieval} and treats organic content, namely the no-ad source such as Wikipedia pages or outputs from internal tools, as a reference that anchors output quality~\citep{guu2020retrieval, izacard2021leveraging}. We incorporate this organic content directly into the social welfare objective, measuring quality as the semantic similarity to no-ad responses. This perspective fundamentally changes the role of the auction:  instead of allocating over the full set of candidate ads, the mechanism admits only those that improve the overall quality-adjusted social welfare.

A key implication of this formulation is an endogenous \emph{reserve price} that screens advertisers based on their marginal contribution relative to the organic baseline. 
This screening step excludes ads that do not contribute positively to the overall quality-adjusted welfare, whether due to semantic irrelevance or low bids, thereby preserving content quality while ensuring a revenue lower bound for the platform.
Compared to prior RAG-based mechanisms~\citep{hajiaghayi2024ad}, which incorporate ads without filtering, our approach explicitly enforces quality preservation at the level of the allocation rule.

Building on the screened RAG model, we design both single-allocation and multi-allocation auction mechanisms. In the single-allocation setting, we maximize a Kullback-Leibler (KL)-regularized social welfare objective that penalizes deviation from  the screened RAG model, which we treat as the semantically optimal retrieval distribution, yielding a tractable allocation rule with payments derived via Myerson’s envelope formula~\citep{myerson1981optimal}, and thus achieving dominant-strategy incentive compatibility (DSIC) and individual rationality (IR). In the multi-allocation setting, we adapt the Vickrey–Clarke–Groves (VCG) mechanism~\citep{clarke1971multipart, groves1973incentives, vickrey1961counterspeculation} to select a welfare-optimal subset of ads. Crucially, the screening step reduces the candidate set to a small eligible subset, making the mechanism computationally efficient while preserving incentive guarantees.

We evaluate our framework across diverse queries, advertiser sets, and bid profiles, and compare against the most closely related RAG-based baseline~\citep{hajiaghayi2024ad}. Across all scenarios, our mechanisms consistently achieve higher revenue,  relevance, and social welfare, while maintaining substantially better output quality, measured as semantic similarity to no-ad responses. 
These results demonstrate that the proposed mechanism design establishes a new paradigm for LLM advertising that enables monetization without compromising output quality.
In conclusion, our contributions are threefold.

\begin{itemize}[leftmargin=*, itemsep=4pt, parsep=0pt, topsep=0pt]
\item We introduce a quality-preserving auction framework for LLM advertising that incorporates organic content into the social welfare objective and derives an endogenous reserve price to screen out ads with negative contributions to the quality-adjusted social welfare.
\item We design incentive-compatible single- and multi-allocation mechanisms: a KL-regularized Myerson mechanism for single allocation, and a screened VCG mechanism for multi-allocation with reduced computational complexity.
\item We provide extensive empirical validation demonstrating consistent improvements over existing RAG-based mechanisms in revenue, semantic relevance, social welfare, and output quality.
\end{itemize}

\paragraph{Related Work.}
Our work relates to three threads of literature: (i) auction design for LLM-generated contents, (ii) RAG, and (iii) reserve prices in auction theory.

First, mechanism design for LLM-generated content has recently emerged at the intersection of machine learning and economics~\citep{feizi2023online}. Existing approaches differ mainly in the \emph{granularity of allocation} and how ads are integrated into generation. At the \emph{token level}, \citet{duetting2024mechanism} propose per-token bidding, but such fine-grained control does not ensure semantic coherence. At the \emph{response level}, \citet{soumalias2024truthful} aggregate multiple candidates at substantial inference cost, while \citet{dubey2024auctions} focus on ad summarization. \citet{balseiro2026position} extend position auctions to AI-generated content under known click-through rates. Most closely related are \emph{segment-level} RAG-based mechanisms by \citet{hajiaghayi2024ad}, which integrate ads without additional inference cost but restrict generation to sponsored content, potentially introducing irrelevant ads, and lack incentive guarantees in the multi-allocation setting. In contrast, we incorporate \emph{organic content} into the allocation problem, yielding a quality-preserving objective that screens out ads with negative marginal welfare contributions and supports incentive-compatible mechanisms.

Second, RAG~\citep{lewis2020retrieval} enhances LLMs by conditioning outputs on retrieved documents and is widely used to improve factuality and relevance in question answering~\citep{guu2020retrieval, izacard2021leveraging, zhao2026retrieval}. Prior work uses RAG primarily as a modeling tool for incorporating external information. We extend this paradigm by treating organic (non-sponsored) content as a first-class component in the retrieval process and embedding it directly into the mechanism design objective.

Finally, reserve prices are a classical tool for revenue optimization in auctions~\citep{myerson1981optimal, ostrovsky2023reserve}, with substantial follow-up work on learning reserve prices from samples or historical data~\citep{cesa2014regret, cole2014sample, chawla2014mechanism, han2025online}. While these works derive reserves from value distributions or bidding data, our reserve prices arise endogenously from a quality-adjusted welfare formulation: each advertiser is screened based on its marginal contribution relative to the organic baseline, simultaneously preserving output quality and guaranteeing a revenue lower bound.

\section{Preliminaries}\label{sec:preli}

We consider $n$ advertisers competing to place ads in the LLM output. Let $[n]=\{1,\dots,n\}$ and $\mathcal{A}=\{a_1,\dots,a_n\}$ denote the advertisers and their ad contents. Given a query $x$, the LLM generates an output $y = y^{(1:T)}$ in $T$ segments (e.g., tokens, sentences, or paragraphs). At each segment $t$, advertisers submit bids $\vec{b}^{(t)} = (b_1^{(t)},\dots,b_n^{(t)})$, which influence generation. We are given a reference LLM model $\P_{\mathrm{ref}}$ and seek to construct a bid-influenced LLM model $\P_{\mathrm{gen}}$.

\paragraph{RAG with Organic Content.}
We adopt a RAG framework~\citep{lewis2020retrieval}, where content generation conditions on retrieved documents without fine-tuning. For each segment $t$, unlike prior work that restricts retrieval to ads~\citep{hajiaghayi2024ad}, we augment the retrieval set $\mathcal{A}$ with \emph{organic content} $a_0^{(t)}$, representing the no-ad source such as Wikipedia pages or outputs from internal tools~\citep{guu2020retrieval, izacard2021leveraging}. For notational simplicity, we drop the superscript of $a_0^{(t)}$ and write the augmented retrieval set as $\mathcal{A}\cup \{a_0\} = \{a_0,a_1,\ldots,a_n\}$. At segment $t$, conditioned on the query $x$ and previously generated segments $y^{(1:t-1)}$, the model samples
\begin{equation}
\label{RAG_NEW} 
\P_{\text{RAG-Sequence}}(y^{(t)} \mid x, y^{(1:t-1)}) := \sum_{i=0}^{n} q_i^{(t)} \P_{\mathrm{ref}}(y^{(t)} \mid x, y^{(1:t-1)}, a_i),
\end{equation}
where the retrieval weight $q_i^{(t)} = \P_{\eta}(a_i \mid x, y^{(1:t-1)})$ denotes the relevance of content $a_i$ to the current context. This formulation allows the model to generate responses without forcing ad inclusion.

\paragraph{Auction Mechanism}
We consider an auction conducted independently at each generation segment. At segment $t$, the mechanism maps bids $\vec{b}^{(t)} \in \mathbb{R}_{\ge 0}^n$ to an allocation $x_i(\vec{b}^{(t)}) \in [0,1]$ for each $i \in [n]\cup\{0\}$, and a payment $p_i(\vec{b}^{(t)}) \ge 0$ for each advertiser $i \in [n]$. Here, $x_i(\vec{b}^{(t)})$ denotes the probability that content item $a_i$ is selected in the next segment, 
and $p_i(\vec{b}^{(t)})$ is the payment charged to advertiser $i$.

We assume the retrieval weight is calibrated to click-through rates (CTR; e.g., \cite{hajiaghayi2024ad}), which can be achieved via standard calibration methods~\citep{mcmahan2013ad, naeini2015obtaining}.
\begin{assumption}\label{Assmp:1}
Let $E_i^{(t)}$ denote the click event for ad $a_i$ at segment $t$. The click-through rate satisfies $\mathrm{ctr}_i^{(t)} := \mathbb{E}[E_i^{(t)} \mid x, y^{(1:t-1)}] = C^{(t)} q_i^{(t)}$ for some normalization constant $C^{(t)}$.
\end{assumption}

Each advertiser $i$ has a private value $v_i \in \mathbb{R}_{\ge 0}$ per click. The expected utility at segment $t$ is $u_i^{(t)}(\vec{b}^{(t)}) = \mathrm{ctr}_i^{(t)} \cdot v_i \, x_i(\vec{b}^{(t)}) - p_i(\vec{b}^{(t)}).$
The auctioneer seeks to design mechanisms that maximize expected revenue subject to dominant-strategy incentive compatibility (DSIC) and individual rationality (IR).  DSIC requires truthful bidding to be optimal for each advertiser, i.e., $u_i^{(t)}(v_i, \vec{b}^{(t)}_{-i})\geq u_i^{(t)}(b_i^{(t)}, \vec{b}^{(t)}_{-i})$ for every possible bid $b_i^{(t)}\in \R_{\geq 0}$ and competing bids $\vec{b}^{(t)}_{-i}\in \R_{\geq 0}^{n-1}$.  IR requires nonnegative utility from participation, i.e., $u_i^{(t)}(v_i, \vec{b}^{(t)}_{-i})\geq0$ for all $\vec{b}^{(t)}_{-i}\in \R_{\geq 0}^{n-1}$.

\paragraph{Design Goal.}
Our objective is to design \emph{quality-preserving} auction mechanisms that balance ad monetization and output fidelity. We incorporate organic content into the welfare objective and derive an endogenous reserve price that admits only ads with positive marginal social welfare contributions. We further penalize deviation from the semantically optimal retrieval distribution to preserve semantic coherence. Our framework yields DSIC and IR mechanisms for both single- and multi-allocation settings while maintaining computational tractability.

In the remainder of the paper, we analyze the auction mechanism at an arbitrary time step $t \in [T]$. For notational brevity, we suppress the superscript $(t)$ for all variables (i.e., $q_i, w_i, b_i, \vec{b}, u_i,\mathrm{ctr}_i, C$) unless otherwise specified.

\section{Main Result}\label{main_result}

\subsection{Screened RAG Model with Welfare-Improving Ad Selection}\label{sec:ad selection}

Introducing advertisements generates revenue but may distort the generated content, degrading user-perceived quality. This trade-off implies that ads with negative marginal contributions to welfare should be excluded.

Let $f(q_0)$ denote the expected social welfare generated by organic content, where $q_0$ is the retrieval weight of the original content, representing its semantic relevance.
We assume $f(q_0)\geq 0$ and that $f(\cdot)$ is increasing in $q_0$.
One natural choice is a power function, $f(q_0) := \eta q_0^{\beta}$ for $\eta > 0$ and $\beta \in (0,1)$, which is uniquely characterized by a scale-invariant welfare ratio condition (Proposition~\ref{thm:uniqueness} in Appendix~\ref{sec:justify}). The concavity ($\beta < 1$) reflects diminishing marginal returns to relevance~\citep{marshall1961principles}, and mirrors the concave value mappings employed in related auction models~\citep{dubey2024auctions}. Under the RAG model~\eqref{RAG_NEW}, the resulting \emph{quality-adjusted social welfare} is
\begin{equation}\label{SW_rag}
\begin{split}
\mathrm{SW}_{\mathrm{rag}}
&:= 
\sum_{i=1}^n q_i\cdot \mathrm{ctr}_i v_i+ q_0\cdot f(q_0),
\end{split}
\end{equation}
where the first term captures advertisers' value and the second term captures the contribution of organic content. Payments are transfers and thus canceled in the social welfare calculation.

In standard RAG, only the top-$k$ documents are retrieved, prioritizing those that contribute most to the generated response~\citep{lewis2020retrieval}. By analogy, it is natural to restrict the ad set $\mathcal{A}$ to advertisers with positive marginal contributions to social welfare. Eq.~\eqref{SW_rag} suggests interpreting $\mathrm{ctr}_i v_i$ as a \emph{pseudo-virtual value}, analogous to the virtual value in classical auction theory~\citep{krishna2009auction}. This  leads to a natural reserve price rule that screens out advertisers with negative marginal contributions, as formalized below.

\begin{theorem}\label{the:reserve}
Under Assumption~\ref{Assmp:1}, given $\hat f(\cdot):= f(\cdot)/C$, $\{q_i\}_{i=0}^n$, and $\{v_i\}_{i=1}^n$, 
including advertiser~$i$ in the candidate set increases expected social 
welfare $\mathrm{SW}_{\mathrm{rag}}$ if and only if $   v_i \geq \hat f(q_0)/q_i.$
Consequently, the welfare-maximizing eligible set is
$
S = \{i \in [n] : v_i \geq \hat f(q_0)/q_i \}.
$
Equivalently, in a DSIC and IR mechanism, this corresponds to imposing a reserve price $r_i = \hat{f}(q_0)/q_i$ on advertiser $i$.
\end{theorem}

The proof is provided in Appendix~\ref{A0.0}. The reserve price screens out advertisers with low relevance or low bids.
This improves content quality while guaranteeing a revenue lower bound for the platform.
Our reserve price differs from classical designs~\citep[e.g.,][]{myerson1981optimal, cesa2014regret, ostrovsky2023reserve} in two key aspects. First, it is driven by the organic welfare function and ad relevance, rather than the distribution of bidder values. Second, its primary role is to preserve output quality rather than to maximize revenue, reflecting the distinct objectives of LLM-based systems.

After excluding advertisers with negative marginal welfare, we define the eligible set $S = \{ i \in [n] : b_i \geq r_i \}$ and augment it with organic content as $S^* = S \cup \{0\}$. We then renormalize the retrieval weights over $S^*$ by
$\tilde{q}_i := {q_i}/{\sum_{j \in S^*} q_j}$,
for $i \in S^*$, so that $\sum_{i \in S^*} \tilde{q}_i = 1$. Here, $\tilde{q}_i$ represents the updated retrieval weight of source $a_i$.
The resulting \emph{screened} RAG model is
\begin{equation}\label{base_rag}
\P_{\mathrm{rag}}(y^{(t)} \mid x, y^{(1:t-1)}) 
:= \sum_{i\in S^*} 
\tilde{q}_i\,
\P_{\mathrm{ref}}(y^{(t)} \mid x, y^{(1:t-1)}, a_i).
\end{equation}

\subsection{Single-Allocation Auction Design}\label{sec:single}
In this section, we  design a single-allocation mechanism in which exactly one content source is selected at each segment, so that $\sum_{i=0}^n x_i(\vec{b}) = 1$. Under the screened RAG model~\eqref{base_rag}, advertisers outside the eligible set $S^*$ are excluded, and we set $x_i(\vec{b}) = 0$ for all $i \notin S^*$.
Given an allocation rule $\{x_i(\vec{b})\}_{i \in S^*}$, the bid-influenced generation model is defined via linear aggregation:
\begin{equation}\label{RAG}
\P_{\mathrm{gen}}(y^{(t)} \mid x, y^{(1:t-1)})=
\sum_{i\in  S^*} x_i(\vec{b}) \P_{\mathrm{ref}}(y^{(t)} \mid x, y^{(1:t-1)}, a_i).
\end{equation}
This linear aggregation follows prior work~\citep{duetting2024mechanism, hajiaghayi2024ad}, but unlike these formulations, ours explicitly incorporates organic content in addition to filtered advertisements.

\paragraph{Allocation Rule.} 
We design an allocation rule that maximizes expected social welfare while controlling deviation from the screened RAG distribution $\P_{\mathrm{rag}}$ in~\eqref{base_rag}, which serves as the semantically optimal baseline in the absence of bid influence. We measure this deviation using the KL divergence
$D_{\mathrm{KL}}(\P_{\mathrm{gen}}(\cdot \mid x, y^{(1:t-1)}) \,\|\, \P_{\mathrm{rag}}(\cdot \mid x, y^{(1:t-1)})),$
and write it as $D_{\mathrm{KL}}(\P_{\mathrm{gen}} \| \P_{\mathrm{rag}})$.

Given that the mechanism observes bids $\vec{b}$ rather than true values $\vec{v}$, we optimize welfare with respect to $\vec{b}$. The allocation $\vec{x} \in \Delta^{|S^*|}:=\{\vec{x}: x_i\ge 0,\sum_{i\in S^*} x_i=1\}$ is chosen to solve
\begin{equation}\label{post_sw}
  \max_{\vec{x} \in \Delta^{|S^*|}}  \mathrm{SW}_{\mathrm{auction}}(\vec{b}, \vec{x})
    :=\sum_{i\in S}  x_i\cdot \mathrm{ctr}_i b_i 
    + x_0\cdot f(q_0) 
    - \lambda D_{\mathrm{KL}}(\mathbb{P}_{\mathrm{gen}} \|\mathbb{P}_{\mathrm{rag}}),
\end{equation}
where $\lambda > 0$ controls the trade-off between revenue and fidelity. Under DSIC, truthful bidding implies $b_i = v_i$, so the mechanism maximizes the true social welfare penalized by deviation from the baseline.

Computing $D_{\mathrm{KL}}(\P_{\mathrm{gen}} \| \P_{\mathrm{rag}})$ directly is intractable due to the exponential output space. The next result provides a tractable upper bound in terms of the allocation vectors $\vec{x}$ and retrieval weights $\tilde{\vec{q}}$.

\begin{proposition}
\label{theo:2}
Let $\vec{\tilde{q}} = (\tilde{q}_i)_{i \in S^*}$ denote the vector of renormalized retrieval weights of the screened RAG model $\mathbb{P}_{\mathrm{rag}}$ in Eq.~\eqref{base_rag}, and let $\vec{x}=(x_i)_{i\in S^*}$ denote the allocation vector of the generative model $\P_{\mathrm{gen}}$ in Eq.~\eqref{RAG}. Then the following inequality holds:
$ D_{\mathrm{KL}}(\P_{\mathrm{gen}} \,\|\, \P_{\mathrm{rag}}) \le D_{\mathrm{KL}}(\vec{x} \,\|\, \vec{\tilde{q}}). 
$ 
\end{proposition}

The proof is provided in Appendix~\ref{A2.0}. 
Proposition~\ref{theo:2} implies that keeping
$ D_{\mathrm{KL}}(\P_{\mathrm{gen}} \,\|\, \P_{\mathrm{rag}})$ small requires the allocation vector $\vec{x}$ to remain close to the semantic-based retrieval weights $\tilde{\vec{q}}$. In particular, when $\vec{x} = \tilde{\vec{q}}$, we recover $\mathbb{P}_{\mathrm{gen}} = \mathbb{P}_{\mathrm{rag}}$. Under Assumption~\ref{Assmp:1}, $\mathrm{ctr}_i \propto q_i$, so the welfare objective $\sum_{i \in S} x_i \cdot \mathrm{ctr}_i b_i + x_0 \cdot f(q_0)$ is proportional to $\sum_{i \in S} x_i \cdot q_i b_i + x_0 \cdot \hat{f}(q_0)$. Combining this observation with Proposition~\ref{theo:2}, we approximate Problem~\eqref{post_sw} by solving the following optimization problem:
\begin{equation}\label{optmize}
\max_{\vec{x} \in \Delta^{|S^*|}} 
\widetilde{\mathrm{SW}}(\vec{b}, \vec{x}):=
\sum_{i\in S}  x_i\cdot \tilde  q_i b_i + x_0\cdot \tilde f(q_0) 
-  \tilde{\lambda}\cdot D_{\mathrm{KL}}(\vec{x} \,\|\, \vec{\tilde{q}}),
\end{equation}
where $\tilde{q}_i= {q_i}/{\sum_{j \in S^*} q_j}$, $\tilde{f}(\cdot) := \hat{f}(\cdot)/\sum_{j\in S^*} q_j$ is the normalized organic welfare function, and $\tilde{\lambda}$ is a regularization parameter. The allocation $\vec{x}(\vec{b})$ is defined as an optimal solution to Problem~\eqref{optmize}, as characterized in the following theorem.

\begin{theorem}
\label{allocation_rule}
The unique optimal solution to Problem~\eqref{optmize} is
\begin{equation}\label{allocate}
    x_i(\vec{b}) =\argmax_{\vec{x} \in \Delta^{|S^*|}} 
 \  \widetilde{\mathrm{SW}}(\vec b, \vec{x})=
    \frac{
        \tilde{q}_i \exp\left({s_i}/{\tilde{\lambda}}\right)
    }{
        \sum_{j\in S^*}
        \tilde{q}_j \exp({s_j}/{\tilde{\lambda}})
    },
    \qquad i \in S^*,
\end{equation}
where
$s_i = \tilde{q}_i b_i\cdot  \mathbf{1}\{i \neq 0\} + \tilde f(q_0)\cdot \mathbf{1}\{i=0\}.$
\end{theorem}
The proof is given in Appendix~\ref{A3.0}. Theorem~\ref{allocation_rule} shows that an advertiser can increase its selection probability by either bidding higher or providing more relevant content. The regularization parameter $\tilde{\lambda}$ controls how closely the allocation follows the semantic-based retrieval weights: as $\tilde{\lambda} \to \infty$, the allocation converges to the screened RAG retrieval weights, $x_i(\vec{b}) \to \tilde{q}_i$, and thus  $\mathbb{P}_{\mathrm{gen}}$ recovers the screened RAG model  $\mathbb{P}_{\mathrm{rag}}$ without bid influence.

\begin{algorithm}[t!]
\caption{\normalsize Quality-Preserving Single-Allocation Auction at Segment $t$}
\begin{algorithmic}[1]

\State \textbf{Input:} Query $x$; previously generated segments $y^{(1:t-1)}$; ad contents $\mathcal{A}=\{a_i\}_{i=1}^n$ and organic content $a_0$; welfare function $\hat f(\cdot)$; regularization parameter $\tilde{\lambda}>0$; bids $\vec{b}=(b_1,\dots,b_n)$.

\Statex \textbf{// Phase 1: Screened RAG Construction (Section 3.1)}

\State Compute the original retrieval weights $q_i = \P_{\eta}(a_i \mid x, y^{(1:t-1)})$ for all $i=0,\dots,n$.

\State Compute reserve prices $r_i = \hat{f}(q_0)/q_i$ for each $i \in [n]$, and form the eligible set 
$S = \{ i \in [n] : b_i \geq r_i \}$, $S^*=\{0\} \cup S.$

\State Re-normalize the retrieval weights over $S^*$ as
$
\tilde{q}_i = {q_i}/{\sum_{j \in S^*} q_j},
$
and construct the screened RAG model $\mathbb{P}_{\mathrm{rag}}$ according to Eq.~\eqref{base_rag}.

\Statex \textbf{// Phase 2: Auction Mechanism (Section 3.2)}

\State Initialize $x_i(\vec{b}) = 0$ and $\tilde p_i(\vec{b}) = 0$ for all $i \notin  S^*$.

\State Compute the optimal allocation $x_i(\vec{b})$ for all $i \in S^*$ according to Eq.~\eqref{allocate}.

\State Compute per-click payments $\tilde{p}_i(\vec{b})$ for all $i \in S$ according to Eq.~\eqref{eq:payment_rule}.

\State \textbf{Output:} Generation model $\mathbb{P}_{\mathrm{gen}}$ Eq.~\eqref{RAG}; allocations $\{x_i(\vec{b})\}_{i \in  S^*}$; payments $\{\tilde{p}_i(\vec{b})\}_{i \in S}$.

\end{algorithmic}
\label{alg:erlam}
\end{algorithm}

\paragraph{Payment Rule.} 

We define the per-click payment $\tilde{p}_i(\vec{b})$ by $p_i(\vec{b}) = \mathrm{ctr}_i\cdot\tilde{p}_i(\vec{b})$, and the corresponding per-click utility $\tilde{u}_i(\vec{b}) := v_i x_i(\vec{b}) - \tilde{p}_i(\vec{b})$. The expected utility can then be written as
$u_i(\vec{b}) = \mathrm{ctr}_i \cdot \tilde{u}_i(\vec{b})$.
For any bidder $i \notin S$ whose bid is below the reserve price, we set $\tilde{p}_i(\vec{b}) = 0$. In practice, the advertiser pays $\tilde{p}_i(\vec{b}) / x_i(\vec{b})$ only when the ad is both shown and clicked, and pays nothing otherwise.
The payment rule of the mechanism is uniquely determined by the DSIC and IR constraints via Myerson’s envelope formula \citep{myerson1981optimal}.

\begin{theorem}
\label{thm:mechanism_dsic_unique}
Under Assumption~\ref{Assmp:1}, for each bidder
$i \in S$ with reserve price $r_i$, any DSIC and IR mechanism for the RAG allocation rule Eq.~\eqref{allocate} has per-click payment rule
\begin{equation}\label{eq:payment_rule}
\tilde{p}_i(\vec{b})=b_i\bigl(x_i(\vec{b})-1\bigr)+r_i+\frac{\tilde{\lambda}}{\tilde{q}_i}\ln\frac{x_i(\vec{b})}{x_i(r_i,\vec{b}_{-i})},
\end{equation}
up to an additive constant. 
The payment rule Eq.~\eqref{eq:payment_rule} corresponds to  the revenue-maximizing 
payment among all DSIC and IR mechanisms.

\end{theorem}

The proof is given in Appendix~\ref{A1.0}. Eq.~\eqref{eq:payment_rule} yields the revenue-maximizing payment, under which $\tilde{u}_i(\vec{b}) = 0$ when $b_i = r_i$. To guarantee strictly positive utility, the auctioneer can reduce this payment by a constant. For example, decreasing $\tilde{p}_i(\vec{b})$ by $r_i x_i(r_i,\vec{b}_{-i})/2$ ensures $\tilde{u}_i(\vec{b})\geq r_i x_i(r_i,\vec{b}_{-i})/2 > 0$, where the strict inequality follows from $x_i(r_i,\vec{b}_{-i}) > 0$ by Eq.~\eqref{allocate}.

\subsection{Multi-Allocation Auction Design}\label{sec:multi}
So far, we have considered a single-allocation setting, where at most one content source is selected per segment. We now extend the model to allow multiple ads to be selected and inserted within a single segment.
Let $A \subseteq \mathcal{A} \cup \{a_0\}$ denote a subset of the augmented retrieval set, and let $a_A$ represent the concatenated content of $\{a_i\}_{i \in A}$. The corresponding combinatorial extension of the RAG model in Eq.~\eqref{RAG_NEW} is given by
\begin{equation*}
    \mathbb{P}_{\text{RAG-Seq}}(y^{(t)} \mid x, y^{(1:t-1)}) 
    = \sum_{A\subseteq \mathcal{A}\cup\{a_0\}} 
    \mathbb{P}_{\eta}(a_A \mid x, y^{(1:t-1)})
    \mathbb{P}_{\mathrm{ref}}(y^{(t)} \mid x, y^{(1:t-1)}, a_A).
\end{equation*}
Since $a_A$ captures only set-level information, naively concatenating ads  can introduce positional biases. As a result, the click-through rate for each ad $i \in A$ must be estimated separately. Following \citet{hajiaghayi2024ad} and in line with Assumption~\ref{Assmp:1}, we assume that click-through rates are proportional to relative relevance.

\begin{assumption}\label{assp:vcg}
In segment $t$, let $E_{A,i}^{(t)}$ denote the click event for ad $i$ when it is displayed within an ad set $A$. We assume that the click-through rate satisfies
\begin{align*}
    \mathrm{ctr}_{A,i}^{(t)} := \mathbb{E}[E_{A,i}^{(t)} \mid x, y^{(1:t-1)}] = C_{\text{multi}}^{(t)} \, q_{A,i}^{(t)},
\end{align*}
for some normalization constant $C_{\text{multi}}^{(t)} > 0$, where
$
    q_{A,i}^{(t)} := \P_\phi\!\big(a_i \mid x, y^{(1:t-1)}, a_A\big)\,\P_\eta\!\big(a_A \mid x, y^{(1:t-1)}\big),
$
and $\sum_{i \in A} \P_\phi(a_i \mid \cdot) = 1$ for any set $A$.
Here, $\phi$ parameterizes the within-set prominence of item $i$, while $\eta$ captures the overall relevance of the ad set $A$.
\end{assumption}
We drop the superscripts and write $q_{A,i}$, $\mathrm{ctr}_{A,i}$, and $C_{\text{multi}}$ for simplicity. The introduction of $q_{A,i}$ suggests the need for a separate module to estimate set-level relevance; we provide one possible heuristic, following \citet{hajiaghayi2024ad}, in Appendix~\ref{sec:heuristic}.
Given bids $\vec{b}$, for any subset $A \subseteq \mathcal{A} \cup \{a_0\}$, the reported social welfare is defined as
$
    \mathrm{SW}_{\text{multi}}(\vec b, A)=\sum_{i\in A}  [\mathrm{ctr}_{A,i} b_i\cdot \mathbf{1}\{i\neq 0\} +  f(q_{A,0})\cdot \mathbf{1}\{i= 0\}].
$
Under Assumption \ref{assp:vcg}, define $\hat{f}_{\text{multi}}(\cdot) := f(\cdot) / C_{\text{multi}}$. Then the multi-allocation social welfare can be written as
\begin{equation}\label{eq:sw_multi} 
\mathrm{SW}_{\text{multi}}(\vec{b}, A) \propto \widehat{\mathrm{SW}}(\vec{b}, A) := \sum_{i \in A} [q_{A,i}\, b_i \cdot \mathbf{1}\{i \neq 0\} + \hat{f}_{\text{multi}}(q_{A,0}) \cdot \mathbf{1}\{i = 0\}]. 
\end{equation}
Since the ad set $\mathcal{A}$ can be large, evaluating social welfare over all subsets $A \subseteq \mathcal{A} \cup \{a_0\}$ is computationally intractable. Moreover, ads with low relevance or low bids may degrade both output quality and welfare. As in Section~\ref{sec:single}, we address this by imposing a reserve price and retaining only advertisers in
$S_{\text{multi}}=\{ i \in [n] : b_i \geq \hat{f}_{\text{multi}}(q_0)/q_i \}$. This screening step both reduces computational complexity and helps filter out low-quality ads.

We then define the feasible set of configurations as
$\mathcal{A}^* := \{ A : A \subseteq S^*_{\text{multi}} \}$ where $S^*_{\text{multi}} = S_{\text{multi}} \cup \{0\} $,
so that $|\mathcal{A}^*| = 2^{|S^*_{\text{multi}}|}$.
Motivated by the classical VCG mechanism \citep{vickrey1961counterspeculation}, we adapt it to the multi-allocation RAG setting based on the social welfare in Eq.~\eqref{eq:sw_multi}, yielding a quality-preserving multi-allocation auction.

\paragraph{Allocation Rule.}
The mechanism selects the welfare-maximizing set:
\begin{equation}\label{eq:vcg_alloc}
    A^* = \argmax_{A \in \mathcal{A}^*} \widehat{\mathrm{SW}}(\vec b,A).
\end{equation}

\paragraph{Payment Rule.}
For each advertiser $i \in A^*$ with $i \neq 0$, define
$A^*_{-i} := \argmax_{A \in \mathcal{A}^*,\, i \notin A} \widehat{\mathrm{SW}}(\vec b, A)$
as the optimal set excluding advertiser $i$. 
The per-click payment for advertiser $i$ is defined as
\begin{equation}\label{eq:vcg_payment}
    \tilde p_i^{\mathrm{VCG}}(\vec b) := \left(\widehat{\mathrm{SW}}(\vec b, A^*_{-i}) 
    - \Bigl[\widehat{\mathrm{SW}}(\vec b, A^*) - q_{A^*,i}\,b_i\Bigr]\right)/q_{A^*,i}.
\end{equation}
The complete algorithm for the quality-preserving multi-allocation auction is given in Algorithm~\ref{alg:vcg} in the Appendix. The per-impression payment is given by $p_i^{\mathrm{VCG}}(\vec b) := \mathrm{ctr}_{A^*,i} \cdot  \tilde p_i^{\mathrm{VCG}}(\vec b) $.
Advertisers $i \notin A^*$ are not selected and therefore pay nothing, yielding utility $u_i = 0$. For each selected advertiser $i \in A^*$ with $i \neq 0$, the utility is $u_i(\vec{b}) = \mathrm{ctr}_{A^*,i} \cdot (v_i - \tilde{p}_i^{\mathrm{VCG}}(\vec{b}))$.

\begin{theorem}\label{thm:vcg}
    Under Assumption~\ref{assp:vcg}, the multi-allocation auction defined by Eqs.~\eqref{eq:vcg_alloc}--\eqref{eq:vcg_payment}
    satisfies DSIC and IR, and maximizes social welfare under truthful reporting.
\end{theorem}
The proof is given in Appendix~\ref{A5.0}. Theorem~\ref{thm:vcg} highlights a key advantage of our VCG-based mechanism: unlike the multi-allocation segment auction of \citet{hajiaghayi2024ad}, which does not provide formal guarantees of incentive compatibility or welfare optimality, our mechanism satisfies both. In addition, their approach requires the number of ads to be fixed in advance, whereas our mechanism endogenously selects the welfare-maximizing subset, resulting in a more flexible and efficient allocation.

\section{Experiments}\label{sec:exp}
We evaluate the proposed mechanisms across multiple advertiser scenarios. At each segment, we first run the auction to select the winning content, then inject it into the LLM prompt to generate the next segment, continuing the response while naturally incorporating the selected content. In our experiments, we approximate the organic content using LLM-generated concise, factual summaries in a Wikipedia-like style, with each round conditioned on previously generated outputs. Full details on the prompts and LLM are provided in Appendix~\ref{sec:detail}.

\paragraph{Setup.} 
We consider outputs consisting of three sentences. The single-allocation mechanism treats each sentence as a separate segment and conducts one auction per sentence, while the multi-allocation auction treats all three sentences as a single segment and conducts one auction for the entire output.
We evaluate six mechanisms in total. 
The first two are variants of our quality-preserving single-allocation auction (Algorithm~\ref{alg:erlam}): \textbf{QP w/ repl.}, which allows repeated selection of the same ad across segments with replacement, and \textbf{QP w/o repl.}, which requires distinct ads at each segment without replacement. The third is \textbf{QP multi-alloc}, our quality-preserving multi-allocation auction (Algorithm~\ref{alg:vcg}).
The remaining three are baselines from \citet{hajiaghayi2024ad}: \textbf{Seg~w/ repl.} and \textbf{Seg~w/o repl.} are the segment single-allocation auction with and without replacement, and \textbf{Seg~multi-alloc} is their multi-allocation variant.

\paragraph{Scenarios.}
We evaluate our mechanisms across four advertiser scenarios, each repeated over 100 independent trials. Scenario~1 uses the query \emph{``What can I visit on a trip to Hawaii?''} from \citet{hu2025gem}; full ad details are provided in Appendix~\ref{ad:s1}. 
We set the KL regularization parameter to $\tilde{\lambda}=1$ and instantiate the organic welfare function as a power function (justified in Appendix~\ref{sec:justify}), with $\hat f(q_0)=2q_0^{0.8}$ for single allocation and $\hat{f}_{\text{multi}}(q_0)=1.5q_0^{0.8}$ for multi-allocation. 
Scenarios~2--4 follow the setup of \citet{hajiaghayi2024ad} and are described in Appendix~\ref{exp:add}. The computation of retrieval weights $q_i$ follows \citep{hajiaghayi2024ad}, with details in Appendix~\ref{para:retrieval}, while the set-level retrieval weight $q_{A}$ extends their heuristic; see Appendix~\ref{sec:heuristic} for details.

\paragraph{Evaluation Metrics.}
We define the following metrics based on realized auction outcomes.
Let $i^{(t)}$ denote the winner of the $t$-th single-allocation auction, or the $t$-th selected advertiser in the multi-allocation auction, and let $p^{(t)}$ denote the corresponding per-click payment.
All metrics are averaged over 100 trials and reported with standard error (SE).
Retrieval weights $q_i^{(t)}$ are not normalized when computing metrics to ensure a fair comparison (referred to as \emph{relevance} hereafter). For QP multi-alloc, $q^{(t)}_{i^{(t)}}$, $T$, and $\hat{f}(\cdot)$ are instantiated as $q_{A^*, i}$, $ |A^*|$, and $\hat{f}_{\text{multi}}(\cdot)$, respectively.

\begin{itemize}[leftmargin=*, itemsep=4pt, parsep=0pt, topsep=0pt]
    \item $\text{Revenue per Ad} := \frac{1}{|\mathcal{T}_{\text{ad}}|} \sum_{t \in \mathcal{T}_{\text{ad}}} p^{(t)} \cdot q^{(t)}_{i^{(t)}}$, the average pseudo-impression payment over all inserted ads, where $\mathcal{T}_{\text{ad}} \subseteq [T]$ denotes the rounds with ad insertion.
    \item $\text{Social Welfare} := \sum_{t=1}^{T} \left[ v_{i^{(t)}} {q}^{(t)}_{i^{(t)}} \cdot \mathbf{1}\{i^{(t)} \neq 0\} + \hat{f}(q_0^{(t)}) \cdot \mathbf{1}\{i^{(t)} = 0\} \right]$, summing allocative value for winning ads and organic value otherwise.
    \item $\text{Relevance} := \sum_{t=1}^{T} q^{(t)}_{i^{(t)}}$, the total retrieval relevance of allocated content.
    \item $\text{KL Divergence} := \sum_{t=1}^{T} D_{\mathrm{KL}}(\vec{x}^{(t)} \,\|\, \vec{\tilde{q}}^{(t)})$, measuring deviation from the organic baseline; smaller values indicate better semantic coherence.
\end{itemize}

\paragraph{Output Quality.} We adopt the same quality measure as \citet{hajiaghayi2024ad}, which computes the cosine similarity between the no ad responses (detailed in Appendix \ref{sec:no_ad}) and the mechanism-generated output, normalized to $[0,1]$, to assess how much the ad insertion degrades the original content quality.

\subsection{Results}\label{exp:s1}

\paragraph{Metrics Comparison.}
The scenario is shown on the left of Table~\ref{tab:scenario0} with five different ads plus an organic document. The organic document has the highest relevance score (results reported as mean $\pm$ SE over 100 trials), justifying its use as the screening threshold for the eligible ads set. Intuitively, ads such as "SunWing Airlines" and "TropicStay" with high bids and relevance are favorable for inclusion, while ads like "NovaSkin" with low relevance may degrade the output quality if included.

\begin{table}[ht!]
    \caption{Left: advertiser bids and relevance scores, including organic relevance $q_0^{(1)}$. Right: auction outcomes averaged over 100 trials (reported as mean $\pm$ SE), best performance in each metric in \textbf{bold}.}
    \vspace{2mm}
    \label{tab:scenario0}
    \begin{minipage}{.272\linewidth}
      \centering
      \resizebox{\textwidth}{!}{
        \begin{tabular}{ccc}
            \toprule
            Advertiser   & Bid  & $q_i^{(1)}$  \\ \midrule
            SunWing     & $3$  & $0.62$ \\ 
            TropicStay   & $3$  & $0.67$ \\ 
            WanderBite   & $2$  & $0.61$ \\ 
            NovaSkin     & $2$  & $0.49$ \\ 
            GridPower    & $1$  & $0.59$ \\
            Organic      & $-$  & $0.80$ {\scriptsize ($\pm .001$)} \\ \bottomrule
        \end{tabular}
        }
    \end{minipage}%
\begin{minipage}{.73\linewidth}
      \centering
        \resizebox{\textwidth}{!}{
            \begin{tabular}{lccccc}
            \toprule
            Mechanism & Revenue per Ad & Soc.\ Wel. & Relevance & KL Div. & Num.\ Ads \\ \midrule
            QP w/ repl.  & $\textbf{1.64}$ {\scriptsize ($\pm .024$)} & $5.90$ {\scriptsize ($\pm .056$)} & $2.17$ {\scriptsize ($\pm .011$)} & $0.02$ {\scriptsize ($\pm .001$)} & $2.15$ {\scriptsize ($\pm .079$)} \\ 
            QP w/o repl. & $1.63$ {\scriptsize ($\pm .024$)} & $5.42$ {\scriptsize ($\pm .016$)} & $2.12$ {\scriptsize ($\pm .010$)} & $\textbf{0.01}$ {\scriptsize ($\pm .001$)} & $1.67$ {\scriptsize ($\pm .051$)} \\ 
            QP multi-alloc     & $1.07$ {\scriptsize ($\pm .002$)} & $\textbf{5.99}$ {\scriptsize ($\pm .004$)} & $\textbf{2.27}$ {\scriptsize ($\pm .002$)} & $-$ & $2.00$ {\scriptsize ($\pm .000$)} \\ 
            Seg w/ repl.\citep{hajiaghayi2024ad}   & $1.13$ {\scriptsize ($\pm .049$)} & $5.00$ {\scriptsize ($\pm .112$)} & $1.97$ {\scriptsize ($\pm .017$)} & $0.18$ {\scriptsize ($\pm .001$)} & $3.00$ {\scriptsize ($\pm .000$)}\\ 
            Seg w/o repl.\citep{hajiaghayi2024ad} & $1.06$ {\scriptsize ($\pm .042$)} & $4.48$ {\scriptsize ($\pm .062$)} & $1.84$ {\scriptsize ($\pm .011$)} & $0.19$ {\scriptsize ($\pm .004$)} & $3.00$ {\scriptsize ($\pm .000$)}\\ 
            Seg multi-alloc\citep{hajiaghayi2024ad}   & $0.88$ {\scriptsize ($\pm .039$)} & $4.46$ {\scriptsize ($\pm .062$)} & $1.82$ {\scriptsize ($\pm .007$)} & $-$ & $3.00$ {\scriptsize ($\pm .000$)}\\ 
            \bottomrule
            \end{tabular}
        }
    \end{minipage}

\end{table}

As shown in Table~\ref{tab:scenario0}, our QP-based mechanisms consistently outperform the Seg baselines in revenue per ad, social welfare, and relevance, while achieving substantially lower KL divergence. This improvement is driven by the reserve price screening, which filters out ads with low relevance or low bids. In this scenario, two ads are clearly dominant, and the QP mechanisms consistently select these for insertion.

\paragraph{Quality Comparison.}
We measure output quality for each individual sentence for single-allocation mechanisms, as well as cumulatively for the first $k$ sentences, with $k=3$ corresponding to the full response evaluated across all mechanisms.
\begin{figure}[h]
    \centering
    \includegraphics[width=\textwidth]{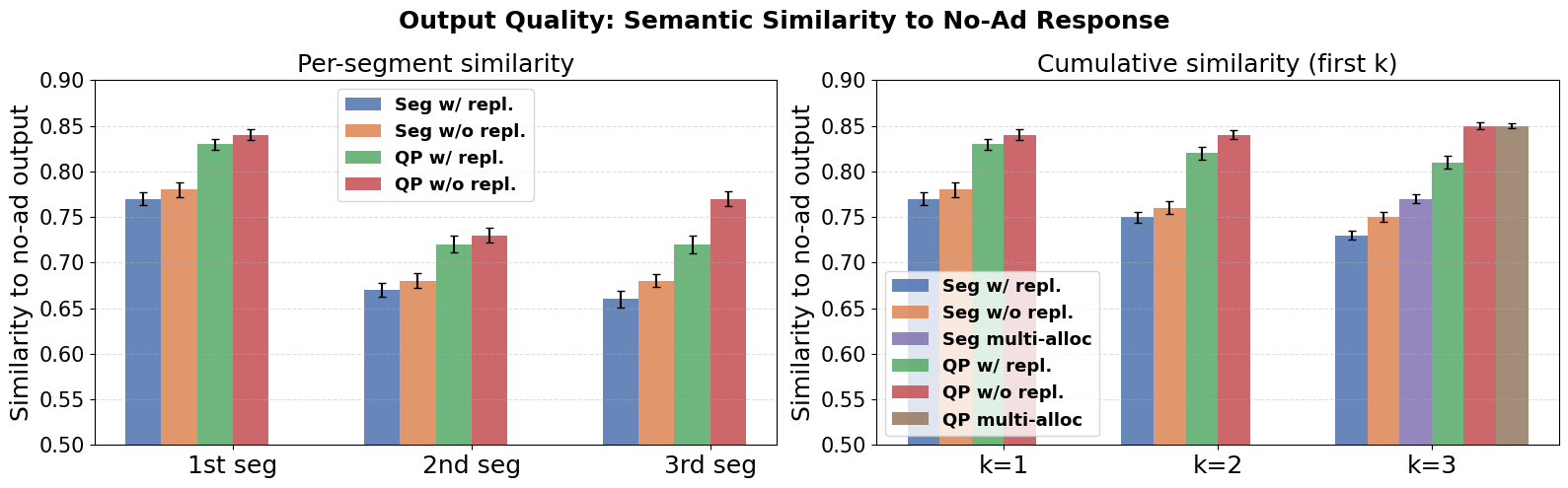}
    \caption{Output quality measured by semantic similarity to the no-ad response, averaged over 100 trials; error bars show $\pm$ SE. Left: per-segment similarity across the three generation rounds. Right: cumulative similarity over the first $k$ sentences.}
    \label{fig:output_quality_0}
\end{figure}

\begin{figure}[h]
    \begin{tcolorbox}[colback=gray!10, colframe=gray!50, width=\textwidth, sharp corners]
    \scriptsize
    \begin{tabular}{p{0.3\textwidth}p{0.3\textwidth}p{0.3\textwidth}}
        \textbf{QP single-allocation (Section \ref{sec:single})}  & \textbf{QP multi-allocation (Section \ref{sec:multi})} & \textbf{Seg single-allocation (See \citep{hajiaghayi2024ad})} \\[2pt]
        \textbf{(Seg 1)} Explore Hawaii's \teal{spectacular volcanoes, lush rainforests, and stunning beaches}, all while enjoying the comfort and local charm of oceanfront villas and jungle hideaways available exclusively through \flink{TropicStay}. 
        &
        \textbf{(Seg 1)}  Embark on your Hawaiian adventure with \flink{SunWing Airlines} and explore the unique ecosystems of Hawai'i \teal{Volcanoes National Park and Haleakala National Park}, where the beauty of untouched nature awaits.
        &
        \textbf{(Seg 1)} Explore Hawaii's stunning landscapes from the \teal{sandy shores of Waikiki to the volcanic peaks of Mauna Kea}, and let \flink{WanderBite} guide you to the islands' most authentic dining experiences, enhancing your adventure with every bite.
        \\[4pt]
        \textbf{(Seg 2)} Take off with \flink{SunWing Airlines}, your trusted partner for seamless travel across the Pacific and beyond, ensuring your journey to Hawaii is as enchanting as the islands themselves.
        &
        \textbf{(Seg 2)} After your flight, settle into a charming, locally-owned property through \flink{TropicStay}, perfectly positioned near cultural landmarks like the \teal{royal Iolani Palace or the vibrant Polynesian Cultural Center}.
        &
        \textbf{(Seg 2)} As you revel in the island's beauty and taste its flavors, protect and nourish your skin with \flink{NovaSkin}, the \red{perfect travel companion that keeps your skin safe and rejuvenated}, no matter where your Hawaiian journey takes you.
        \\[4pt]
        \textbf{(Seg 3)} Beyond the natural wonders, consider immersing yourself in the local history and culture by visiting the \teal{Iolani Palace and the Bishop Museum}, or participating in traditional Hawaiian activities like \teal{luaus and lei-making workshops}.
        &
        \textbf{(Seg 3)} This seamless combination of comfortable travel and authentic local living ensures your Hawaiian getaway is both enriching and memorable.
        &
        \textbf{(Seg 3)} Start your unforgettable Hawaiian experience with \flink{SunWing Airlines}, where comfort meets convenience, ensuring your trip across the Pacific is as pleasurable as your destination.
        \\
    \end{tabular}
    \end{tcolorbox}
    \caption{Sample outputs of different mechanisms. \textcolor{teal}{Teal}: useful answer information.  \textcolor{blue}{Blue}: ads link. \textcolor{red}{Red}: incoherent or forced ad insertion.}
    \label{fig:qual1}
\end{figure}

Figure~\ref{fig:output_quality_0} shows that our mechanisms consistently achieve higher output quality across all segments, highlighting the benefit of incorporating organic content into the generation process. QP w/o repl.\ attains the highest quality overall, especially in the third segment, where previously selected ads are excluded, increasing the likelihood of choosing organic content. Multi-allocation mechanisms also outperform single-allocation ones, consistent with \citet{hajiaghayi2024ad}, as they provide greater flexibility for organizing content across the output. Figure~\ref{fig:qual1} further illustrates these differences through representative examples.

From the qualitative outputs in Figure~\ref{fig:qual1}, QP single-allocation uses the organic document when no ad is selected, producing informative content, as seen in Segment 3. QP multi-allocation does not always guarantee informative ad-free segments (see Segment~3), but it integrates organic and sponsored content more naturally across the full response, resulting in comparable overall quality, consistent with Figure~\ref{fig:output_quality_0}.
In contrast, Seg single-allocation enforces one ad per segment, which can lead to the insertion of low-relevance ads, such as ``NovaSkin'' in Segment 2, disrupting the response flow. As a result, useful information is concentrated in the first segment, while later segments are increasingly dominated by ads. These observations highlight the advantage of our quality-preserving mechanisms. Similar patterns hold in Scenarios 2--4; see Appendix~\ref{exp:add} for details.

\section{Conclusion}\label{sec:conclusion}
We study mechanism design for LLM advertising and propose a quality-preserving auction framework that integrates organic content into the RAG pipeline and uses an endogenous reserve price to admit only welfare-improving ads. On this screened model, we develop a KL-regularized single-allocation mechanism with Myerson payments and a screened VCG multi-allocation mechanism. Both achieve dominant-strategy incentive compatibility and individual rationality. Experiments across different scenarios show consistent gains in revenue per ad, relevance, and social welfare, while preserving semantic fidelity to no-ad responses. These results demonstrate that effective monetization of LLM systems need not compromise output quality.

Several directions remain open for future work. 
First, the organic welfare function can be learned from platform objectives and user interaction data, rather than being specified a priori. This includes estimating how users trade off relevance, informativeness, and intrusiveness, and using these signals to adaptively shape the reserve price and the resulting revenue–quality trade-off across contexts. 
Second, incorporating budget constraints and repeated interactions would better capture real-world dynamics, as advertisers face spending limits and platforms operate over multiple rounds. Finally, developing richer measures of output quality beyond semantic similarity may further improve alignment between monetization and user experience, for example by incorporating factors such as factual accuracy or user engagement.

\newpage

{
\bibliography{auctionLLM}
}

\newpage
\appendix
\normalsize
\section*{Appendix}
\noindent 

\section{Choice of the Organic Social Welfare Function}\label{sec:justify}
We adopt $f(q_0) := \eta q_0^{\beta}$ with $\eta, \beta > 0$ in our experiments. A similar power functional form is adopted in \citet{dubey2024auctions} to map relative prominence to realized click-through rate, serving as a natural bridge between semantic relevance and advertiser value.
The following proposition shows that, under continuity, monotonicity, and a scale-invariant welfare ratio, the organic social welfare function $f$ is uniquely characterized.

\begin{proposition}
\label{thm:uniqueness}
Let $f: (0, 1] \to \mathbb{R}_{\geq 0}$ satisfy:
\begin{enumerate}[label=\textup{(A\arabic*)}, leftmargin=*, nosep]
    \item \textbf{Monotonicity}: $f$ is strictly increasing.
    \item \textbf{Relative Welfare Invariance}: For all $q, q' \in (0,1]$ and $\eta > 0$ with $\eta q, \eta q' \in (0,1]$, $f(\eta q)/f(\eta q') = f(q)/f(q')$.
    \item \textbf{Continuity}: $f$ is continuous on $(0,1]$.
\end{enumerate}
Then $f(q_0) = \eta q_0^{\beta}$ for some $\eta, \beta > 0$ is the unique family satisfying \textup{(A1)--(A3)}.
\end{proposition}
\begin{proof}
Rewrite \textup{(A2)} as
\[
    f(\eta q) = \frac{f(\eta q')}{f(q')} \cdot f(q)
    \quad \forall\, \eta, q, q' \in (0,1].
\]
Since the left-hand side is independent of $q'$, the ratio 
$f(\eta q') / f(q')$ must be constant in $q'$; denote it 
$g(\eta)$. We thus obtain the multiplicative Cauchy functional 
equation:
\[
    f(\eta q) = g(\eta)\cdot f(q)
    \quad \forall\, \eta, q \in (0,1].
\]
Setting $q = 1$ gives $f(\eta) = g(\eta) f(1)$, so $g(\eta) = f(\eta)/f(1)$. Defining $\hat{f}(q) := f(q)/f(1)$, 
we get $\hat{f}(\eta q) = \hat{f}(\eta)\hat{f}(q)$, 
the standard multiplicative Cauchy equation with $\hat{f}(1) = 1$. 
By \textup{(A3)}, the unique continuous solution is 
$\hat{f}(q) = q^{\beta}$ for some $\beta \in \mathbb{R}$ 
\citep{aczel2006lectures}, hence $f(q) = \eta q^{\beta}$ 
where $\eta := f(1) > 0$ by \textup{(A1)}. Finally, \textup{(A1)} requires 
$\beta > 0$, since $\eta q^{\beta}$ is strictly increasing 
on $(0,1]$ if and only if $\beta > 0$. Conversely, $f(q) = \eta q^{\beta}$ with $\eta, \beta > 0$ clearly satisfies \textup{(A1)--(A3)}. This completes the proof of Proposition~\ref{thm:uniqueness}.
\end{proof}

The parameter $\beta \in (0,1)$ further captures \emph{diminishing marginal returns} to organic content quality: as $q_0$ increases, each incremental gain in relevance yields a progressively smaller increase in social welfare. This concavity is consistent with the classical law of diminishing marginal utility~\citep{marshall1961principles} and the concave value mappings employed in related auction models~\citep{dubey2024auctions}. Intuitively, when organic content is already highly relevant, users are well-served and marginal welfare gains diminish; conversely, low-relevance content leaves substantial room for improvement. We therefore restrict attention to $\beta \in (0,1)$ throughout.

\section{A Heuristic for the Set-Level Retrieval Weight $q_A$}\label{sec:heuristic}

One possible heuristic, following \citet{hajiaghayi2024ad}, is 
\[
q_A = \sum_{i \in A} q_i + \frac{\xi}{|A|(|A|-1)} \sum_{i \neq j \in A} s(i,j) \cdot \mathrm{rel}(a_i, a_j),
\]
with per-item attribution given by $q_{A,i} = \frac{q_i}{\sum_{j \in A} q_j} q_A$. 
Here, $\mathrm{rel}(a_i, a_j)$ measures the semantic relevance between $a_i$ and $a_j$, which can be computed using the semantic similarity introduced in Appendix~\ref{para:retrieval}. Moreover, $s(i,j)=+1$ if $i=0$ or $j=0$, and $s(i,j)=-1$ otherwise, thereby rewarding ads that are complementary to the organic content while penalizing mutually redundant ones. The parameter $\xi \geq 0$ controls the strength of the pairwise term, and the second term vanishes when $|A|=1$.
The normalization factor $|A|(|A|-1)$ corresponds to the total number of ordered pairs, ensuring that the pairwise term is invariant to the set size. Since \citet{hajiaghayi2024ad} fixes $|A|$ ex ante, this factor reduces to a constant and can be absorbed into $\xi$, and is therefore not explicitly included in their formulation. In our experiments, we choose $\xi=1$, which balances the individual relevance terms and the pairwise complementarity term.

\section{Full Algorithm of the QP Multi-Allocation Auction Mechanism}
\label{sec:vcg_alg}
For completeness, Algorithm~\ref{alg:vcg} presents the full pseudocode of the quality-preserving multi-allocation auction introduced in Section~\ref{sec:multi}. The algorithm proceeds in two phases: an eligibility screening phase that filters out ads failing to meet the reserve price condition, followed by a VCG allocation and payment computation phase.

\begin{algorithm}[h]
\caption{\normalsize Quality-Preserving Multi-Allocation Auction at Segment $t$}
\begin{algorithmic}[1]

\State \textbf{Input:} Query $x$; previously generated segments $y^{(1:t-1)}$;  ad contents $\mathcal{A}=\{a_i\}_{i=1}^n$ and organic content $a_0$;  welfare function $\hat{f}_{\text{multi}}(\cdot)$; bids $\vec{b}=(b_1,\dots,b_n)$.

\Statex \textbf{// Phase 1: Eligible Set Construction}

\State Compute retrieval weights $q_i = \P_{\eta}(a_i \mid x, y^{(1:t-1)})$ 
for all $i = 0, \dots, n$, and form the eligible set
\[
S_{\text{multi}} = \{ i \in [n] : b_i \geq \hat{f}_{\text{multi}}(q_0)/q_i \},  \quad S^*_{\text{multi}} = S_{\text{multi}} \cup \{0\} ,\quad 
\mathcal{A}^* = \{ A : A \subseteq S^*_{\text{multi}} \}.
\]

\State Compute $q_{A,i}$ for all $A \in \mathcal{A}^*$ and $i \in A$
(e.g., via the heuristic in Section~\ref{sec:heuristic}).

\Statex \textbf{// Phase 2: Multi-Allocation Auction Mechanism (Section \ref{sec:multi})}

\State Compute the optimal allocation according to Eq.~\eqref{eq:vcg_alloc}:
$
A^* = \argmax_{A \in \mathcal{A}^*} \widehat{\mathrm{SW}}(\vec{b}, A).
$

\State For each $i \in A^*$ with $i \neq 0$, compute per-click payment 
$\tilde{p}_i^{\mathrm{VCG}}(\vec{b})$ according to Eq.~\eqref{eq:vcg_payment}. 
Set $\tilde{p}_i^{\mathrm{VCG}}(\vec{b}) = 0$ for all advertisers $i \notin A^*$.

\State \textbf{Output:} Winning set $A^*$; per-click payments 
$\{\tilde{p}_i^{\mathrm{VCG}}(\vec{b})\}_{i \in A^*, i\neq 0}$.

\end{algorithmic}
\label{alg:vcg}
\end{algorithm}

\section{Proofs of Main Results}\label{app:proofs}

\subsection{Proof of Theorem \ref{the:reserve}}\label{A0.0}

\begin{proof}
Since $\sum_{i=0}^n q_i = 1$, the social welfare in Eq.~\eqref{SW_rag} can be 
rewritten as
\begin{equation*}
\mathrm{SW}_{\mathrm{rag}}
= f(q_0) + \sum_{i=1}^n q_i \big(\mathrm{ctr}_i v_i - f(q_0)\big).
\end{equation*}
Since $q_i > 0$ for all $i$ and each advertiser can be independently included in or excluded from the candidate set, including advertiser $i$ increases social welfare if and only if $\mathrm{ctr}_i v_i - f(q_0) \geq 0$, and the welfare-maximizing eligible set is $S = \{i \in [n] : \mathrm{ctr}_i v_i \geq f(q_0)\}$.

Advertiser $i$ should thus be selected only if $v_i \geq f(q_0)/\mathrm{ctr}_i$.
Under Assumption~\ref{Assmp:1}, $\mathrm{ctr}_i = Cq_i$, which yields
\[
r_i := \frac{f(q_0)}{\mathrm{ctr}_i} = \frac{f(q_0)}{Cq_i} = \frac{\hat{f}(q_0)}{q_i}.
\]
Under a DSIC and IR mechanism, truthful reporting $b_i = v_i$ is a 
dominant strategy, and the above threshold is implementable via the 
reserve price $r_i$. This completes the proof of Theorem \ref{the:reserve}.
\end{proof}

\subsection{Proof of Proposition \ref{theo:2}}\label{A2.0}
\begin{proof}
For notational brevity, we omit the conditioning on $(x, y^{(1:t-1)})$ 
and simply write distributions over $y^{(t)}$. 
Let $\P_i := \P_{\mathrm{ref}}(\cdot \mid a_i)$ for $i \in S^*$.

By definitions Eq.~\eqref{RAG} and Eq.~\eqref{base_rag}, both $\P_{\mathrm{gen}}$ 
and $\P_{\mathrm{rag}}$ are mixtures over the same family $\{\P_i\}_{i\in S^*}$:
\[
    \P_{\mathrm{gen}} = \sum_{i\in S^*} x_i\, \P_i, 
    \qquad 
    \P_{\mathrm{rag}} = \sum_{i\in S^*} \tilde{q}_i\, \P_i.
\]

Let $I$ be a random index taking values in $S^*$, and define joint distributions over $S^* \times \mathcal{Y}$:
\[
    P_{IY}(i,\, y) := x_i\,\P_i(y), 
    \qquad 
    Q_{IY}(i,\, y) := \tilde{q}_i\,\P_i(y).
\]
The marginals on $\mathcal{Y}$ are exactly $\P_{\mathrm{gen}}$ and 
$\P_{\mathrm{rag}}$, respectively. 
Computing the KL divergence between the joint distributions:
\begin{align*}
    D_{\mathrm{KL}}(P_{IY} \,\|\, Q_{IY}) 
    &= \sum_{i\in S^*} \sum_{y} x_i\,\P_i(y) 
       \log \frac{x_i\,\P_i(y)}{\tilde{q}_i\,\P_i(y)} \\
    &= \sum_{i\in S^*} x_i \log \frac{x_i}{\tilde{q}_i} 
       \underbrace{\sum_{y} \P_i(y)}_{=\,1} \\
    &= D_{\mathrm{KL}}\!\left(\vec{x} \,\|\, \vec{\tilde{q}}\right).
\end{align*}
By the data processing inequality (marginalization over $I$ cannot increase KL divergence), we have that
\begin{align*}
    D_{\mathrm{KL}}(\P_{\mathrm{gen}} \,\|\, \P_{\mathrm{rag}}) 
    &= D_{\mathrm{KL}}\!\left(\sum_{i\in S^*} P_{IY}(i,\cdot) 
       \,\middle\|\, \sum_{i\in S^*} Q_{IY}(i,\cdot)\right) \\
    &\le D_{\mathrm{KL}}(P_{IY} \,\|\, Q_{IY}) \\
    &= D_{\mathrm{KL}}\!\left(\vec{x} \,\|\, \vec{\tilde{q}}\right).
\end{align*}

This completes the proof of Proposition~\ref{theo:2}.
\end{proof}

\subsection{Proof of Theorem \ref{allocation_rule}}\label{A3.0}
\begin{proof}
The feasible set $\Delta^{|S^*|}:=\{\vec{x}: x_i\ge 0,\sum_{i\in S^*} x_i=1\}$ is compact and convex. The objective is
strictly concave in $\vec{x}$: the linear terms are concave, and
$\tilde{\lambda}D_{\mathrm{KL}}(\vec{x}\|\vec{\tilde{q}})$ is strictly 
convex for $\tilde{\lambda}>0$. Hence a unique maximizer exists, and it lies in 
the relative interior of $\Delta^{|S^*|}$ since the KL term diverges on the 
boundary.

Introduce a Lagrange multiplier $\lambda$ for the constraint
$\sum_{i\in S^*}x_i = 1$ and form
\[
\mathcal{L}(\vec{x},\lambda)
= \sum_{i\in S^*} s_i\, x_i
  - \tilde{\lambda}\sum_{i\in S^*} x_i \ln\frac{x_i}{\tilde{q}_i}
  - \lambda\!\left(\sum_{i\in S^*} x_i - 1\right),
\]
where we have used the compact form $s_i x_i$ to unify the two cases 
in the objective (since $s_0 x_0(\vec{b}) = x_0(\vec{b})\cdot \tilde f(q_0)$ 
and $s_i x_i = x_i\cdot \tilde{q}_i v_i $ for $i\in S$).
Setting $\partial\mathcal{L}/\partial x_i = 0$ for each 
$i \in S^*$ gives
\[
    s_i - \tilde{\lambda}\!\left(\ln x_i - \ln\tilde{q}_i + 1\right) 
    - \lambda = 0,
\]
which rearranges to
\[
    \ln x_i = \frac{s_i}{\tilde{\lambda}} + \ln\tilde{q}_i 
    - 1 - \frac{\lambda}{\tilde{\lambda}}.
\]

Exponentiating yields
\[
    x_i = \tilde{q}_i
          \exp\!\left(\frac{s_i}{\tilde{\lambda}}\right)
          \cdot e^{-(1+\lambda/\tilde{\lambda})},
    \qquad i \in S^*.
\]
The constant $e^{-(1+\lambda/\tilde{\lambda})}$ is fixed by the normalization
constraint $\sum_{i\in S^*}x_i = 1$:
\[
    e^{-(1+\lambda/\tilde{\lambda})}
    = \frac{1}{\sum_{j=0}^{m}
        \tilde{q}_j\exp \left(\tfrac{s_j}{\tilde{\lambda}}\right)}.
\]
Substituting back gives the stated formula. Uniqueness follows from the
strict concavity of the objective.
This completes the proof of Theorem \ref{allocation_rule}.
\end{proof}

\subsection{Proof of Theorem \ref{thm:mechanism_dsic_unique}}\label{A1.0}
\begin{proof}
We prove the result using Myerson's characterization of truthful mechanisms
\citep{myerson1981optimal}. Throughout, we fix $\vec{b}_{-i}$ and view the allocation and payment as functions of advertiser $i$'s bid.

\paragraph{Step 1: Monotonicity of the Allocation.}
For bidder $i \in S$, the allocation rule Eq.~\eqref{allocate} at bid $b_i$ is
\[
x_i(b_i, \vec{b}_{-i})
= \frac{\tilde q_i \exp(\tilde q_i b_i / \tilde \lambda)}
       {\tilde q_i \exp(\tilde q_i b_i / \tilde \lambda) + R_{-i}},
\qquad
R_{-i} = \sum_{j \neq i} \tilde q_j \exp(s_j / \tilde \lambda).
\]
Differentiating with respect to $b_i$ yields
\[
\frac{\partial x_i}{\partial b_i}
= \frac{\tilde q_i^2}{\tilde \lambda}\,
x_i(\vec b) \bigl(1 - x_i(\vec b)\bigr) > 0.
\]
Hence the allocation rule is strictly increasing in $b_i$.

\paragraph{Step 2: Envelope Condition and DSIC.}
Consider truthful reporting $b_i = v_i$, under which the per-click utility is $\tilde u_i(v_i, \vec{b}_{-i}) = v_i\, x_i(v_i, \vec{b}_{-i}) - \tilde p_i(v_i, \vec{b}_{-i})$. By Myerson's lemma, any DSIC mechanism with the monotone allocation rule established in Step~1 must satisfy the envelope condition
\[
\frac{\partial \tilde u_i(v_i, \vec{b}_{-i})}{\partial v_i} = x_i(v_i, \vec{b}_{-i}).
\]
Integrating from $r_i$ to $v_i$ gives
\[
\tilde u_i(v_i, \vec{b}_{-i}) 
= \tilde u_i(r_i, \vec{b}_{-i}) 
+ \int_{r_i}^{v_i} x_i(z, \vec{b}_{-i})\, dz.
\]
Solving for $\tilde p_i$ from the definition of $\tilde u_i$ yields
\begin{equation}\label{eq:myerson_payment}
\tilde p_i(v_i, \vec{b}_{-i})
= v_i\, x_i(v_i, \vec{b}_{-i}) 
- \tilde u_i(r_i, \vec{b}_{-i})
- \int_{r_i}^{v_i} x_i(z, \vec{b}_{-i})\, dz,
\end{equation}
where $\tilde u_i(r_i, \vec{b}_{-i})$ is an additive constant in $v_i$.

\paragraph{Step 3: Closed Form.}
Evaluating the integral yields
\[
\int_{r_i}^{v_i} x_i(z, \vec{b}_{-i})\, dz
= v_i - r_i 
- \frac{\tilde \lambda}{\tilde q_i} 
\ln \frac{x_i(v_i, \vec{b}_{-i})}{x_i(r_i, \vec{b}_{-i})}.
\]
Substituting into Eq.~\eqref{eq:myerson_payment} gives
\[
\tilde p_i(v_i, \vec{b}_{-i})
= v_i \bigl(x_i(v_i, \vec{b}_{-i}) - 1\bigr) 
+ r_i 
+ \frac{\tilde \lambda}{\tilde q_i} 
\ln \frac{x_i(v_i, \vec{b}_{-i})}{x_i(r_i, \vec{b}_{-i})}
- \tilde u_i(r_i, \vec{b}_{-i}).
\]
Since this identity must hold for every $v_i \in \mathbb{R}_{\geq 0}$ and $\tilde p_i$ is defined as a function of the reported bid, renaming the argument $v_i$ as $b_i$ yields Eq.~\eqref{eq:payment_rule} up to the additive constant $\tilde u_i(r_i, \vec{b}_{-i})$.

\paragraph{Step 4: IR.}
Since $x_i(b_i, \vec{b}_{-i})$ is strictly increasing in $b_i$ (Step~1), we have $x_i(\vec b) \geq x_i(r_i, \vec{b}_{-i})$ for all $b_i \geq r_i$. Therefore
\[
\tilde u_i(\vec b) 
= \tilde u_i(r_i, \vec{b}_{-i}) 
+ \int_{r_i}^{b_i} x_i(z, \vec{b}_{-i})\, dz 
\geq \tilde u_i(r_i, \vec{b}_{-i}).
\]
The IR constraint $\tilde u_i(\vec b) \geq 0$ is satisfied if and only if $\tilde u_i(r_i, \vec{b}_{-i}) \geq 0$, i.e.\ $\tilde p_i(r_i, \vec{b}_{-i}) \leq r_i\, x_i(r_i, \vec{b}_{-i})$.
The payment rule Eq.~\eqref{eq:payment_rule} corresponds to the canonical choice $\tilde u_i(r_i, \vec{b}_{-i}) = 0$, which is the revenue-maximizing selection among all DSIC and IR mechanisms.
This completes the proof of Theorem~\ref{thm:mechanism_dsic_unique}.
\end{proof}

\subsection{Proof of Theorem \ref{thm:vcg}}\label{A5.0}

\begin{proof}

\textbf{(i) IR.} We verify $u_i(v_i, \vec{b}_{-i}) \geq 0$ for all 
$\vec{b}_{-i} \in \mathbb{R}_{\geq 0}^{n-1}$.

Under truthful reporting $b_i = v_i$, the mechanism selects $A^* = \argmax_{A \in \mathcal{A}^*} \widehat{\mathrm{SW}}((v_i, \vec{b}_{-i}), A)$. If $i \notin A^*$, then $p_i^{\mathrm{VCG}}(v_i, \vec{b}_{-i}) = 0$ and $u_i(v_i, \vec{b}_{-i}) = 0 \geq 0$ trivially. If $i \in A^*$, then by the payment rule Eq.~\eqref{eq:vcg_payment},
\begin{align*}
    u_i(v_i, \vec{b}_{-i}) 
    &= \mathrm{ctr}_{A^*,i} \cdot v_i - p_i^{\mathrm{VCG}}(v_i, \vec{b}_{-i}) \\
    &= C_{\text{multi}} \cdot \bigl(q_{A^*,i}\, v_i - \widehat{\mathrm{SW}}((v_i, \vec{b}_{-i}), A^*_{-i}) + \widehat{\mathrm{SW}}((v_i, \vec{b}_{-i}), A^*) - q_{A^*,i}\, v_i \bigr) \\
    &= C_{\text{multi}} \cdot \bigl(\widehat{\mathrm{SW}}((v_i, \vec{b}_{-i}), A^*) - \widehat{\mathrm{SW}}((v_i, \vec{b}_{-i}), A^*_{-i})\bigr) \\
    &\geq 0,
\end{align*}
where the last inequality follows since $A^*$ maximizes $\widehat{\mathrm{SW}}((v_i, \vec{b}_{-i}), \cdot)$ over all $A \in \mathcal{A}^*$, including those with $i \notin A$, so $\widehat{\mathrm{SW}}((v_i, \vec{b}_{-i}), A^*) \geq \widehat{\mathrm{SW}}((v_i, \vec{b}_{-i}), A^*_{-i})$.

\textbf{(ii) DSIC.} We verify $u_i(v_i, \vec{b}_{-i}) \geq u_i(b_i', \vec{b}_{-i})$ for every 
$b_i' \in \mathbb{R}_{\geq 0}$ and $\vec{b}_{-i} \in \mathbb{R}_{\geq 0}^{n-1}$.

Under truthful reporting, the calculation in (i) shows
\[
u_i(v_i, \vec{b}_{-i}) = C_{\text{multi}} \cdot \bigl(\widehat{\mathrm{SW}}((v_i, \vec{b}_{-i}), A^*) - \widehat{\mathrm{SW}}((v_i, \vec{b}_{-i}), A^*_{-i})\bigr),
\]
where this identity also holds when $i \notin A^*$ since in that case $A^* = A^*_{-i}$ and both sides equal zero.

Under a misreport $b_i'$, the mechanism selects $\hat{A} = \argmax_{A \in \mathcal{A}^*} \widehat{\mathrm{SW}}((b_i', \vec{b}_{-i}), A)$. If $i \in \hat{A}$,
\begin{align*}
    u_i(b_i', \vec{b}_{-i}) 
    &= \mathrm{ctr}_{\hat{A},i}\, v_i - p_i^{\mathrm{VCG}}(b_i', \vec{b}_{-i}) \\
    &= C_{\text{multi}} \cdot \bigl(\widehat{\mathrm{SW}}((b_i', \vec{b}_{-i}), \hat{A}) - \widehat{\mathrm{SW}}((b_i', \vec{b}_{-i}), A^*_{-i}) - q_{\hat{A},i}(b_i' - v_i)\bigr) \\
    &= C_{\text{multi}} \cdot \bigl(\widehat{\mathrm{SW}}((v_i, \vec{b}_{-i}), \hat{A}) - \widehat{\mathrm{SW}}((v_i, \vec{b}_{-i}), A^*_{-i})\bigr),
\end{align*}
where the last equality uses the fact that $\widehat{\mathrm{SW}}((b_i', \vec{b}_{-i}), \hat{A}) - q_{\hat{A}, i} b_i' = \widehat{\mathrm{SW}}((v_i, \vec{b}_{-i}), \hat{A}) - q_{\hat{A}, i} v_i$ for $i \in \hat{A}$, and $\widehat{\mathrm{SW}}(\cdot, A^*_{-i})$ does not depend on $i$'s bid since $i \notin A^*_{-i}$.

If $i \notin \hat{A}$, then $p_i^{\mathrm{VCG}}(b_i', \vec{b}_{-i}) = 0$ and $u_i(b_i', \vec{b}_{-i}) = 0$, which is no greater than $u_i(v_i, \vec{b}_{-i})$ by (i).

In the case $i \in \hat{A}$, since $A^*$ maximizes $\widehat{\mathrm{SW}}((v_i, \vec{b}_{-i}), \cdot)$ over $\mathcal{A}^*$, we have $\widehat{\mathrm{SW}}((v_i, \vec{b}_{-i}), A^*) \geq \widehat{\mathrm{SW}}((v_i, \vec{b}_{-i}), \hat{A})$, hence
\[
u_i(v_i, \vec{b}_{-i}) - u_i(b_i', \vec{b}_{-i}) = C_{\text{multi}} \cdot \bigl(\widehat{\mathrm{SW}}((v_i, \vec{b}_{-i}), A^*) - \widehat{\mathrm{SW}}((v_i, \vec{b}_{-i}), \hat{A})\bigr) \geq 0.
\]

Under truthful reporting $\vec{b} = \vec{v}$, the VCG mechanism selects $A^* = \argmax_{A \in \mathcal{A}^*} \widehat{\mathrm{SW}}(\vec{v}, A)$ by Eq.~\eqref{eq:vcg_alloc}, which directly maximizes social welfare. This completes the proof of Theorem \ref{thm:vcg}.
\end{proof}

\section{Implementation Details}\label{sec:detail}

Section \ref{sec:prompt} presents the prompt templates used across different auction mechanisms, and Section \ref{LLM:api} describes the GPT-4 API configuration used as the LLM backbone in all experiments.

\subsection{Prompts}\label{sec:prompt}
This section details the prompts used across different mechanisms. For the compared methods (i.e., Seg w/ repl., Seg w/o repl., Seg multi-alloc), we adopt the prompts from \citet{hajiaghayi2024ad} without modification to ensure reproducibility.

\paragraph{Initial Query.}
For Scenario 1, we use the following initial query from \citet{hu2025gem}:
\begin{lstlisting}[language=Python]
prompt = "What can I visit on a trip to Hawaii?"
\end{lstlisting}
For Scenarios 2--4, we use the following initial query following \citet{hajiaghayi2024ad}:
\begin{lstlisting}[language=Python]
prompt = "Can you suggest some books similar to \"To Kill a Mockingbird\"?"
\end{lstlisting}

\paragraph{Prompt to generate the organic content $a_0$.}
At round $k=1$, the organic content $a_0$ is generated by prompting the model to produce a concise, factual three-sentence summary in a Wikipedia-like style. For $k > 1$, the prompt is conditioned on the previous cumulative output, instructing the model to introduce new information without repeating prior content, thereby ensuring coherence across rounds.
\begin{lstlisting}[language=Python]
# Round k = 1 (no prior output)
query = f"""For the following question, write a summarized answer
        of 3 sentences that provides factual, neutral context
        relevant to the topic. Write like a Wikipedia excerpt.
        Question: {prompt}"""

# Round k > 1 (conditioned on previous output)
query = f"""For the following question, write a summarized answer
        of 3 sentences that introduces new relevant information
        not yet mentioned in the partial answer, providing factual,
        neutral context to help continue the answer.
        Do NOT repeat or comment on anything already mentioned
        in the partial answer. Write like a Wikipedia excerpt.
        Question: {prompt}
        Partial answer so far: {previous_output}"""
\end{lstlisting}

\paragraph{Quality-preserving single-allocation auction (Section \ref{sec:single}).}
At each round $k$, we select the winning content to generate a new sentence. The query is conditioned on whether the source is organic or sponsored (ad), and whether a prior output exists.
\begin{lstlisting}[language=Python]

# Round k = 1, if using organic content {organic_doc} (no ad)
query = f"""{prompt}
        Please respond to this question in only one sentence.
        You may use the following background context to inform your answer:
        >> {organic_doc}"""

# Round k = 1, if using ad content {ad_text}
query = f"""{prompt}
        Please respond to this question in only one sentence while also
        advertise {advertiser_name} with this context >> {ad_text}
        Make sure to connect the answer and the advertisement very naturally,
        not something like appending the ads after just answering the question.
        Focus on answering the question, there shouldn't be too much
        advertisement in the output.
        Make sure the response is one sentence."""

# Round k > 1, if using organic content {organic_doc} (no ad)
query = f"""You must continue your answer to my original query: {prompt}
        Your previous response was >> {previous_output}
        Please add exactly one new sentence at the end that continues naturally
        without modifying, paraphrasing, or removing any part of your
        previous response.
        You may use the following content to inform the new sentence:
        >> {organic_doc}
        Do NOT mention or reference any advertisement in the new sentence.
        Make sure there is exactly one new sentence added.
        Write the entire document, which merges your previous response
        and the new sentence."""

# Round k > 1, if using ad content {ad_text}
query = f"""You must continue your answer to my original query: {prompt}
        Your previous response was >> {previous_output}
        And you now should advertise {advertiser_name}, but without hurting
        the coherency of the entire document.
        Here's some context about {advertiser_name}:
        >> {ad_text}
        Make sure there is exactly one new sentence added.
        Write the entire document, which merges your previous response
        and the new sentence."""
\end{lstlisting}

\paragraph{Quality-preserving multi-allocation auction (Section \ref{sec:multi}).}
The VCG auction allocates a winning set $A^*$ of advertisers, and all sentences are generated in a single call. The prompt varies depending on the number of winning advertisers $|A^*|$ and whether the organic content $a_0$ is included in $A^*$.
\begin{lstlisting}[language=Python]

# No winning advertiser (organic only)
query = f"""{prompt}
        Please respond to this question in exactly 3 sentences.
        You may use the following background context to inform your answer:
        >> {organic_doc}
        Do NOT mention or reference any advertisement.
        Make sure the response is exactly 3 sentences."""

# Exactly one winning advertiser, a_0 not in A*
query = f"""{prompt}
        Please respond to this question in exactly 3 sentences
        while also advertising {advertiser_name} with this context >>
        {ad_text}
        Make sure to connect the answer and the advertisement very naturally,
        not something like appending the ads after just answering the question.
        Focus on answering the question, there shouldn't be too much
        advertisement in the output.
        Mention {advertiser_name} exactly only once.
        Make sure the response is exactly 3 sentences."""

# Exactly one winning advertiser, organic content a_0 in A*
query = f"""{prompt}
        Please respond to this question in exactly 3 sentences
        while also advertising {advertiser_name} with this context >>
        {ad_text}
        You may also use the following background context:
        >> {organic_doc}
        Make sure to connect the answer and the advertisement very naturally,
        not something like appending the ads after just answering the question.
        Focus on answering the question, there shouldn't be too much
        advertisement in the output.
        Mention {advertiser_name} exactly only once.
        Make sure the response is exactly 3 sentences."""

# Multiple winning advertisers, organic content a_0 not in A*, suppose |A*|=m
query = f"""{prompt}
        Please respond to this question in exactly 3 sentences
        while
        (1) advertise {advertisers[0]} with this context >>
        {ads[0]}
        (2) advertise {advertisers[1]} with this context >>
        {ads[1]}
        ...
        (m) advertise {advertisers[m-1]} with this context >>
        {ads[m-1]}
        Make sure to connect the answer and the advertisement very naturally,
        not something like appending the ads after just answering the question.
        Focus on answering the question, there shouldn't be too much
        advertisement in the output.
        Make sure to mention all brands ({all_ads_names}), each brand mentioned
        exactly only once, and ensure that the response is exactly
        3 sentences."""

# Multiple winning advertisers, a_0 in A*, suppose |A*|=m+1
query = f"""{prompt}
        Please respond to this question in exactly 3 sentences
        while
        (1) advertise {advertisers[0]} with this context >>
        {ads[0]}
        (2) advertise {advertisers[1]} with this context >>
        {ads[1]}
        ...
        (m) advertise {advertisers[m-1]} with this context >>
        {ads[m-1]}
        You may also use the following background context:
        >> {organic_doc}
        Make sure to connect the answer and the advertisement very naturally,
        not something like appending the ads after just answering the question.
        Focus on answering the question, there shouldn't be too much
        advertisement in the output.
        Make sure to mention all brands ({all_ads_names}), each brand mentioned
        exactly only once, and ensure that the response is exactly
        3 sentences."""
\end{lstlisting}

\subsection{LLM API Configuration}\label{LLM:api}
To ensure a fair comparison with \citet{hajiaghayi2024ad}, we adopt the same GPT-4 API configuration as in their work. Each API call consists of a single user message containing the query. 
\label{app:prompt-conf}
\begin{lstlisting}[language=Python]
messages = [{"role": "user", "content": query}]
response = client.chat.completions.create(
        model = "gpt-4-turbo",
        temperature = 1,
        max_tokens = 300,  # 500 for organic content generation
        messages = messages,)
\end{lstlisting}

\subsection{Retrieval Weight}\label{para:retrieval}
Following \citet{hajiaghayi2024ad}, we compute retrieval weights $q_i^{(t)}$ 
using the \texttt{sentence-transformers} library 
(\url{https://huggingface.co/sentence-transformers/multi-qa-MiniLM-L6-cos-v1}), 
which encodes text into dense vector representations such that semantically similar texts are mapped to nearby vectors,
as measured by cosine similarity.

\section{Ad Document Content}\label{sec:ad doc}

\subsection{Synthetic Advertisers for Scenario 1}\label{ad:s1}

The following synthetic advertisers and their corresponding ad copies were generated by an LLM, used as the ad set in Scenario 1 (Section~\ref{exp:s1}). Each entry corresponds to a fictional company and begins with a one-sentence description, followed by the document content used in the prompt during output generation.

\begin{enumerate}
    \item \flink{SunWing Airlines}: A Pacific-focused airline offering direct routes to tropical destinations with premium in-flight service and flexible booking options.
        \begin{tcolorbox}[colback=gray!10, colframe=gray!50, width=\textwidth, sharp corners]
        Take off with SunWing Airlines, your trusted partner for seamless travel across the Pacific and beyond. With spacious cabins, award-winning in-flight service, and direct routes to the world's most breathtaking destinations, SunWing makes every journey as memorable as the destination itself. Book today and enjoy exclusive deals, flexible cancellation, and a loyalty rewards program that turns miles into memories. Your next adventure starts here — fly with SunWing.
        \end{tcolorbox}
    \item \flink{TropicStay}: A vacation rental platform connecting travelers with handpicked, locally-owned properties in the world's most beautiful destinations.
        \begin{tcolorbox}[colback=gray!10, colframe=gray!50, width=\textwidth, sharp corners]
        Discover the perfect retreat with TropicStay, the premier platform for booking vacation rentals in the world's most beautiful destinations. From oceanfront villas to cozy jungle hideaways, TropicStay connects travelers with handpicked, locally-owned properties that offer authentic experiences you won't find anywhere else. With instant booking, 24/7 guest support, and verified reviews from real travelers, planning your dream getaway has never been easier. Explore, book, and belong — only with TropicStay.
        \end{tcolorbox}
    \item \flink{WanderBite}: A food discovery app that connects travelers with the best local restaurants and hidden dining gems around the globe.
        \begin{tcolorbox}[colback=gray!10, colframe=gray!50, width=\textwidth, sharp corners]
        Unlock the culinary soul of every destination with WanderBite, the app that connects food lovers to the best local restaurants, street food spots, and hidden dining gems around the globe. Whether you're craving fresh poke bowls on the North Shore or a Michelin-starred tasting menu in the city, WanderBite guides you there with curated recommendations, real-time reservations, and exclusive foodie deals. Because the best travel memories are made at the table. Download WanderBite and taste the world.
        \end{tcolorbox}
    \item \flink{NovaSkin}: A dermatologist-approved skincare brand offering reef-safe, SPF-focused products designed for active, outdoor lifestyles.
        \begin{tcolorbox}[colback=gray!10, colframe=gray!50, width=\textwidth, sharp corners]
        Introducing NovaSkin, the dermatologist-approved skincare line engineered for the modern lifestyle. From SPF-50 daily moisturizers to overnight repair serums, NovaSkin's lightweight, reef-safe formulas protect and restore your skin whether you're under the office lights or the open sun. Trusted by over 10 million customers worldwide, NovaSkin combines cutting-edge biotechnology with clean, sustainable ingredients to deliver visible results you can feel confident about. Because great skin doesn't take a vacation — and neither does NovaSkin.
        \end{tcolorbox}
    \item \flink{GridPower Bank}: A portable charging brand producing ultra-slim, waterproof power banks built for travelers and outdoor adventurers.
        \begin{tcolorbox}[colback=gray!10, colframe=gray!50, width=\textwidth, sharp corners]
        Stay charged through every adventure with GridPower Bank, the ultra-slim, high-capacity portable charger built for life on the move. Featuring rapid-charge technology, dual USB-C ports, and a rugged waterproof design, GridPower Bank keeps your devices powered through long flights, beach days, and everything in between. Compact enough to fit in your pocket, powerful enough to charge your laptop twice over — GridPower Bank is the travel essential you didn't know you needed. Power up and go.
        \end{tcolorbox}
\end{enumerate}

\subsection{Ads Used in Scenarios 2--4}\label{ad:s2-4}
To facilitate direct comparison with \citet{hajiaghayi2024ad}, we adopt their exact ad setup for Scenarios 2--4. Each entry follows the same format as in Appendix~\ref{ad:s1}, beginning with a one-sentence description followed by the full document content. For completeness, we reproduce the one-sentence description of each ad document below; full descriptions can be found in \citet{hajiaghayi2024ad}. Scenarios 2 and 3 (Appendices~\ref{app:sce2} and~\ref{sec:exp2}) use the first four ads, while Scenario 4 (Appendix~\ref{sec:exp3}) uses all eleven ads.

\begin{enumerate}
    \item \flink{Velora}: A tech company that designs and sells premium, seamlessly integrated smart devices and services for a sophisticated and efficient lifestyle.
    \item \flink{BookHaven}: An online bookstore offering a vast selection of books across all genres with a seamless shopping experience and reliable delivery.
    \item \flink{MassMart}: A membership-based retail store offering premium bulk products at unbeatable prices with a focus on customer satisfaction and community support.
    
    \item \flink{EspressoEdge}: A premium coffee shop offering high-quality, handcrafted beverages made from the finest Arabica beans, providing a luxurious coffee experience for all.
    \item \flink{SocialHub}: A leading social media platform that connects over two billion users through personalized news feeds, interactive groups, and tools for sharing life's moments and promoting businesses.
    
    \item \flink{ColaBubbles}: The world’s favorite soft drink, known for its unique flavor blend and effervescent bubbles that have been delighting people for over a century.
    \item \flink{FizzyPop}: An iconic soft drink celebrated for its crisp, refreshing flavor and vibrant effervescence, perfect for those who live life boldly and seek excitement in every moment.
    \item \flink{SkyTech}: The world’s leading aerospace company, designing and manufacturing advanced commercial airplanes, defense systems, and space technologies to ensure safe and efficient global connectivity and exploration.
    \item \flink{AeroDynamics}: The global leader in aerospace innovation, designing and manufacturing advanced commercial aircraft that provide unparalleled comfort, efficiency, and reliability for a superior flying experience.
    \item \flink{MusicStream}: The ultimate destination for streaming millions of songs with personalized recommendations and offline listening capabilities, offering a seamless music experience anytime, anywhere.
    \item \flink{BrainChips}: The global leader in semiconductor technology, providing cutting-edge processors that power a wide range of devices with industry-leading performance, reliability, and security for professionals, gamers, and more.
\end{enumerate}

\section{No-Ad Response}\label{sec:no_ad}
To measure output quality, we use semantic similarity between the ad-augmented output and the no-ad response as our baseline signal. For each trial, the similarity score is computed by comparing its output against each of the ten no-ad responses individually and then averaging the results. 
\subsection{No-Ad Responses for Scenario 1}
We list the ten no-ad baseline responses used in Scenario 1 (Figure~\ref{fig:output_quality_0}), where the user query is:\indent
    \textit{``What can I visit on a trip to Hawaii?''}
Each response comprises three sentences and contains no allocated advertisements.

\begin{tcolorbox}[colback=gray!10, colframe=gray!50, width=\textwidth, sharp corners]
\textbf{Response 1.} On a trip to Hawaii, you can visit the stunning beaches of Waikiki on Oahu, famous for their golden sands and surfing opportunities. Explore the lush landscapes and dramatic waterfalls of the Road to Hana in Maui. Additionally, don't miss the chance to experience the awe-inspiring volcanoes at Hawaii Volcanoes National Park on the Big Island, where you can witness active lava flows and rich geological history.
\end{tcolorbox}

\begin{tcolorbox}[colback=gray!10, colframe=gray!50, width=\textwidth, sharp corners]
\textbf{Response 2.} On a trip to Hawaii, you can explore the stunning landscapes of Hawaii Volcanoes National Park, home to the active Kilauea and Mauna Loa volcanoes. Enjoy the beautiful beaches, such as Waikiki Beach on Oahu or Hapuna Beach on the Big Island, perfect for sunbathing, swimming, and surfing. Additionally, you can experience traditional Hawaiian culture at a luau, where you can enjoy local food, music, and hula dancing.
\end{tcolorbox}

\begin{tcolorbox}[colback=gray!10, colframe=gray!50, width=\textwidth, sharp corners]
\textbf{Response 3.} Hawaii offers a rich array of attractions, including the stunning beaches of Waikiki on Oahu and the fascinating Hawaii Volcanoes National Park on the Big Island, where you can observe active lava flows and volcanic landscapes. Nature enthusiasts can explore the lush landscapes and dramatic cliffs of the Napali Coast on Kauai via hiking trails or boat tours. Additionally, the historical significance of Pearl Harbor and the USS Arizona Memorial on Oahu offers a poignant look at an important WWII site.
\end{tcolorbox}

\begin{tcolorbox}[colback=gray!10, colframe=gray!50, width=\textwidth, sharp corners]
\textbf{Response 4.} On a trip to Hawaii, you can explore the stunning landscapes of Volcanoes National Park on the Big Island, where you can witness active lava flows and rare geological formations. Discover the rich history and culture of the islands by visiting the Iolani Palace and the Bishop Museum in Honolulu, Oahu. Additionally, enjoy the breathtaking views and hiking opportunities along the Na Pali Coast on Kauai or relax on the pristine beaches of Maui like Kaanapali Beach or Wailea Beach.
\end{tcolorbox}

\begin{tcolorbox}[colback=gray!10, colframe=gray!50, width=\textwidth, sharp corners]
\textbf{Response 5.} On a trip to Hawaii, you can explore the stunning landscapes of Volcanoes National Park on the Big Island, where you can witness active lava flows and learn about the region's geothermal phenomena. Dive into Hawaii's rich culture at the Polynesian Cultural Center on Oahu, offering insights into the traditions, dance, and history of the various Pacific Islander communities. Don't miss the chance to relax on the world-famous beaches of Maui or engage in water activities like snorkeling and surfing in the crystal-clear ocean waters.
\end{tcolorbox}

\begin{tcolorbox}[colback=gray!10, colframe=gray!50, width=\textwidth, sharp corners]
\textbf{Response 6.} On a trip to Hawaii, you can explore the stunning beaches of Waikiki on Oahu, known for their golden sands and excellent surfing conditions. You might also visit the Hawaii Volcanoes National Park on the Big Island, where you can witness the awe-inspiring sight of active lava flows and volcanic landscapes. Additionally, the historic Pearl Harbor site, also located on Oahu, offers a profound look at an important WWII site, including the USS Arizona Memorial.
\end{tcolorbox}

\begin{tcolorbox}[colback=gray!10, colframe=gray!50, width=\textwidth, sharp corners]
\textbf{Response 7.} On a trip to Hawaii, you can explore the stunning landscapes of Hawaii Volcanoes National Park, home to the active Kilauea volcano. Enjoy the beautiful beaches of Waikiki on Oahu and take surfing lessons in the clear blue waters. Additionally, experience the rich Hawaiian culture and history at the Polynesian Cultural Center or by visiting the historic Iolani Palace in Honolulu.
\end{tcolorbox}

\begin{tcolorbox}[colback=gray!10, colframe=gray!50, width=\textwidth, sharp corners]
\textbf{Response 8.} On a trip to Hawaii, you can explore the stunning landscapes of Volcanoes National Park on the island of Hawaii, where you can witness active lava flows and learn about the region's geology. Be sure to visit the beautiful beaches of Maui, such as Wailea Beach and Kaanapali Beach, for sunbathing, snorkeling, and water sports. Additionally, take time to immerse yourself in Hawaiian culture and history at the Polynesian Cultural Center on Oahu or by attending a traditional luau with Hawaiian music, dance, and food.
\end{tcolorbox}

\begin{tcolorbox}[colback=gray!10, colframe=gray!50, width=\textwidth, sharp corners]
\textbf{Response 9.} In Hawaii, you can explore the stunning landscapes of Volcanoes National Park on the Big Island, where Kilauea and Mauna Loa volcanoes are located. Enjoy the beautiful beaches and water activities like snorkeling and surfing on islands like Maui and Oahu. Don't miss the opportunity to immerse yourself in Hawaiian culture and history at cultural sites like the Polynesian Cultural Center or the historic Iolani Palace in Honolulu.
\end{tcolorbox}

\begin{tcolorbox}[colback=gray!10, colframe=gray!50, width=\textwidth, sharp corners]
\textbf{Response 10.} In Hawaii, you can explore the stunning landscapes of Volcanoes National Park on the Big Island, where Kilauea and Mauna Loa, two of the world's most active volcanoes, reside. On Oahu, visit the historic Pearl Harbor site, including the USS Arizona Memorial, which pays homage to those lost in the 1941 attacks. Don't miss the chance to experience the breathtaking beauty of the Na Pali Coast in Kauai either by boat tour, helicopter, or hiking the Kalalau Trail.
\end{tcolorbox}

\subsection{No-Ad Responses for Scenarios 2--4}\label{noad:2-4}
We list the ten no-ad baseline responses used in Scenarios 2--4 (Figures~\ref{fig:output_quality}, \ref{fig:output_quality_2}, and \ref{fig:output_quality_3}), where the user query is
 \textit{``Can you suggest some books similar to `To Kill a Mockingbird'?''}
Each response comprises three sentences and contains no allocated advertisements.

\begin{tcolorbox}[colback=gray!10, colframe=gray!50, width=\textwidth, sharp corners]
\textbf{Response 1.} Certainly! If you enjoyed ``To Kill a Mockingbird'' for its explorations of social injustices and moral growth, you might enjoy ``The Help'' by Kathryn Stockett, which delves into the lives of African American maids working in white households in Mississippi during the 1960s. Another powerful book is ``A Time to Kill'' by John Grisham, which tackles themes of racial tension and justice in a small Southern town. Additionally, ``Roll of Thunder, Hear My Cry'' by Mildred D.\ Taylor, set during the Great Depression, offers a poignant narrative about a Black family navigating racism and economic hardship.
\end{tcolorbox}

\begin{tcolorbox}[colback=gray!10, colframe=gray!50, width=\textwidth, sharp corners]
\textbf{Response 2.} Certainly! For books similar to ``To Kill a Mockingbird,'' you might enjoy ``Go Set a Watchman'' by Harper Lee, which delves further into the complexities of morality and social justice in the South. Another recommendation is ``The Secret Life of Bees'' by Sue Monk Kidd, which tackles themes of race, love, and redemption in the 1960s South. Lastly, ``Roll of Thunder, Hear My Cry'' by Mildred D.\ Taylor explores racism and injustice through the eyes of a young black girl in the Great Depression-era South.
\end{tcolorbox}

\begin{tcolorbox}[colback=gray!10, colframe=gray!50, width=\textwidth, sharp corners]
\textbf{Response 3.} Certainly! If you enjoyed ``To Kill a Mockingbird'' by Harper Lee, you might like ``The Help'' by Kathryn Stockett, which also explores themes of race and morality through vividly drawn characters and a gripping narrative. Another recommendation is ``Go Set a Watchman,'' also by Harper Lee, which continues the story of the Finch family and discusses evolving societal norms in the South during the 1950s. Lastly, ``A Lesson Before Dying'' by Ernest J.\ Gaines provides a poignant look at justice and dignity in a racially divided community, much like Lee's iconic novel.
\end{tcolorbox}

\begin{tcolorbox}[colback=gray!10, colframe=gray!50, width=\textwidth, sharp corners]
\textbf{Response 4.} Certainly! If you enjoyed ``To Kill a Mockingbird,'' you might like ``The Help'' by Kathryn Stockett, which deals with themes of racial inequality and social justice in the 1960s South. Another recommendation is ``A Time to Kill'' by John Grisham, which also explores racial tensions and legal battles in a Southern town. Lastly, ``Go Set a Watchman'' by Harper Lee, the controversial follow-up to ``To Kill a Mockingbird,'' provides a different perspective on some of the same characters and themes.
\end{tcolorbox}

\begin{tcolorbox}[colback=gray!10, colframe=gray!50, width=\textwidth, sharp corners]
\textbf{Response 5.} Certainly, if you enjoyed ``To Kill a Mockingbird'' for its themes of social justice and moral growth, you might like ``The Help'' by Kathryn Stockett, which explores racism and civil rights in 1960s Mississippi through the perspectives of both black maids and the white women they work for. Another poignant novel is ``A Time to Kill'' by John Grisham, which also tackles issues of racism and justice in the American South. Additionally, Harper Lee's ``Go Set a Watchman,'' although controversial, offers more insight into the characters and setting of her famed novel and presents an older Scout grappling with her perceptions of her father and her upbringing.
\end{tcolorbox}

\begin{tcolorbox}[colback=gray!10, colframe=gray!50, width=\textwidth, sharp corners]
\textbf{Response 6.} Certainly! If you enjoyed ``To Kill a Mockingbird'' by Harper Lee, you might also like ``The Help'' by Kathryn Stockett, which explores themes of racism and social justice in the American South during the 1960s through the eyes of both black maids and the white women they work for. Another great choice is ``A Time to Kill'' by John Grisham, which also tackles issues of racial tension and justice in the South but with a gripping legal thriller twist. Lastly, ``Go Set a Watchman,'' also by Harper Lee, is a controversial but intriguing follow-up to her first novel, featuring many of the same characters dealing with evolving societal norms in the 1950s.
\end{tcolorbox}

\begin{tcolorbox}[colback=gray!10, colframe=gray!50, width=\textwidth, sharp corners]
\textbf{Response 7.} Certainly! If you enjoyed ``To Kill a Mockingbird,'' you might like ``The Help'' by Kathryn Stockett, which similarly tackles themes of racial injustice and morality but in 1960s Mississippi. Another book to consider is ``Go Set a Watchman'' by Harper Lee, which is actually a sequel to ``To Kill a Mockingbird'' and explores some of the same characters and themes in a later timeframe. Additionally, ``Roll of Thunder, Hear My Cry'' by Mildred D.\ Taylor, set in the Depression-era South, also delves into themes of racism and social justice through the perspective of a young African American girl.
\end{tcolorbox}

\begin{tcolorbox}[colback=gray!10, colframe=gray!50, width=\textwidth, sharp corners]
\textbf{Response 8.} Certainly! If you enjoyed ``To Kill a Mockingbird'' by Harper Lee, you might like ``The Secret Life of Bees'' by Sue Monk Kidd, which deals with themes of race, loss, and the search for identity in a similar coming-of-age format. Another great option is ``Go Set a Watchman,'' also by Harper Lee, which is a complex and somewhat controversial follow-up involving some of the same characters from her iconic debut. Additionally, ``A Time to Kill'' by John Grisham addresses themes of racial tension and justice in the South, similar to Lee's portrayal in her novel.
\end{tcolorbox}

\begin{tcolorbox}[colback=gray!10, colframe=gray!50, width=\textwidth, sharp corners]
\textbf{Response 9.} If you appreciated ``To Kill a Mockingbird'' for its powerful exploration of social justice and morality, you might enjoy ``The Help'' by Kathryn Stockett, which delves into the complexities of race relations in the 1960s South through the eyes of both black maids and the white women they work for. Another book to consider is ``A Time to Kill'' by John Grisham, which revolves around a tense legal battle after a black man takes justice into his own hands following an atrocious crime against his daughter in a racially charged Southern town. For a broader historical perspective, you might like ``Roll of Thunder, Hear My Cry'' by Mildred D.\ Taylor, which provides the viewpoint of an African American girl growing up in Depression-era Mississippi, grappling with racism and social injustice.
\end{tcolorbox}

\begin{tcolorbox}[colback=gray!10, colframe=gray!50, width=\textwidth, sharp corners]
\textbf{Response 10.} Certainly! If you enjoyed ``To Kill a Mockingbird,'' you might like ``The Help'' by Kathryn Stockett, which explores themes of racial tension and justice in the South through the lives of African American maids working for white families. Another recommendation is ``A Time to Kill'' by John Grisham, which also tackles issues of racial injustice and moral dilemmas in a small Southern town. Lastly, ``Go Set a Watchman'' by Harper Lee, the controversial follow-up to ``To Kill a Mockingbird,'' provides a different perspective on some of the same characters and themes.
\end{tcolorbox}

\section{Additional Experimental Results}\label{exp:add}

In this section, we present three additional experimental scenarios beyond Scenario~1 in Section~\ref{sec:exp}. Due to the absence of established standard benchmarks, all experiments in Scenarios~2--4 use the setup of \citet{hajiaghayi2024ad} to ensure fair comparison. Each scenario uses the query \emph{``Can you suggest some books similar to `To Kill a Mockingbird'?''} with advertiser information detailed in Appendix~\ref{ad:s2-4}; the corresponding no-ad baseline responses are provided in Appendix~\ref{noad:2-4}.
The metrics and output quality are computed in the same manner as Scenario 1.

\subsection{Scenario 2}\label{app:sce2}

\paragraph{Metrics Comparison.} We consider a scenario with four advertisers whose bids and initial retrieval weights $q_i^{(1)}$ are listed in the left side of Table~\ref{tab:scenario1}.
We use the same hyperparameter settings as Scenario~1: $\tilde{\lambda}=1$, $\hat{f}(q_0)=2q_0^{0.8}$ for single-allocation, and $\hat{f}_{\text{multi}}(q_0)=1.5q_0^{0.8}$ for multi-allocation.

\begin{table}[h]
    \caption{Metrics comparison results across mechanisms for Scenario 2. Left: advertiser bids and relevance scores, including organic relevance $q_0^{(1)}$. Right: auction outcomes averaged over 100 trials (reported as mean $\pm$ SE), best performance in each metric in \textbf{bold}.}
    \vspace{2mm}
    \begin{minipage}{.31\linewidth}
      \centering
      \resizebox{\textwidth}{!}{
        \begin{tabular}{ccc}
            \toprule
            Advertiser   & Bid  & $q_i^{(1)}$  \\ \midrule
            Velora       & $3$  & $0.53$ \\ 
            BookHaven    & $3$  & $0.70$ \\ 
            MassMart     & $2$  & $0.51$ \\ 
            EspressoEdge & $2$  & $0.49$ \\ 
            Organic      & $-$  & $0.78$ {\scriptsize ($\pm .003$)} \\ \bottomrule
        \end{tabular}
        }
    \end{minipage}%
    \begin{minipage}{.695\linewidth}
      \centering
        \resizebox{\textwidth}{!}{
            \begin{tabular}{lccccc}
            \toprule
            Mechanism & Revenue per Ad & Soc.\ Wel. & Relevance & KL Div. & Num.\ Ads \\ \midrule
            QP w/ repl.  & $\textbf{1.58}$ {\scriptsize ($\pm .028$)} & $5.80$ {\scriptsize ($\pm .057$)} & $2.15$ {\scriptsize ($\pm .012$)} & $0.04$ {\scriptsize ($\pm .001$)} & $2.03$ {\scriptsize ($\pm .085$)} \\ 
            QP w/o repl. & $1.53$ {\scriptsize ($\pm .042$)} & $5.23$ {\scriptsize ($\pm .015$)} & $2.11$ {\scriptsize ($\pm .014$)} & $\textbf{0.02}$ {\scriptsize ($\pm .001$)} & $1.36$ {\scriptsize ($\pm .061$)} \\ 
            QP multi-alloc     & $1.01$ {\scriptsize ($\pm .001$)} & $\textbf{5.82}$ {\scriptsize ($\pm .004$)} & $\textbf{2.21}$ {\scriptsize ($\pm .002$)} & $-$ & $2.00$ {\scriptsize ($\pm .000$)} \\ 
            Seg w/ repl.\citep{hajiaghayi2024ad}   & $1.15$ {\scriptsize ($\pm .045$)} & $5.01$ {\scriptsize ($\pm .087$)} & $1.86$ {\scriptsize ($\pm .017$)} & $0.06$ {\scriptsize ($\pm .000$)} & $3.00$ {\scriptsize ($\pm .000$)} \\ 
            Seg w/o repl.\citep{hajiaghayi2024ad} & $1.07$ {\scriptsize ($\pm .041$)} & $4.62$ {\scriptsize ($\pm .037$)} & $1.75$ {\scriptsize ($\pm .007$)} & $0.05$ {\scriptsize ($\pm .001$)} & $3.00$ {\scriptsize ($\pm .000$)} \\ 
            Seg multi-alloc\citep{hajiaghayi2024ad}   & $0.82$ {\scriptsize ($\pm .034$)} & $4.47$ {\scriptsize ($\pm .040$)} & $1.70$ {\scriptsize ($\pm .007$)} & $-$ & $3.00$ {\scriptsize ($\pm .000$)} \\ 
            \bottomrule
            \end{tabular}
        }
    \end{minipage}

    \label{tab:scenario1}
\end{table}

The organic relevance $q_0^{(1)}$ is the highest among all documents, justifying its use as the screening threshold for the eligible ads set. Notably, ``BookHaven'' and ``Velora'' exhibit relatively high relevance and bids, whereas ``MassMart'' and ``EspressoEdge'' have low relevance, suggesting that including them may degrade output quality. The right side of Table \ref{tab:scenario1} shows consistent results with Table \ref{tab:scenario0}, further verifying the efficiency of our QP-based methods.

\paragraph{Quality Comparison.}
We measure output quality in Figure \ref{fig:output_quality}.

\begin{figure}[ht!]
    \centering
    \includegraphics[width=\textwidth]{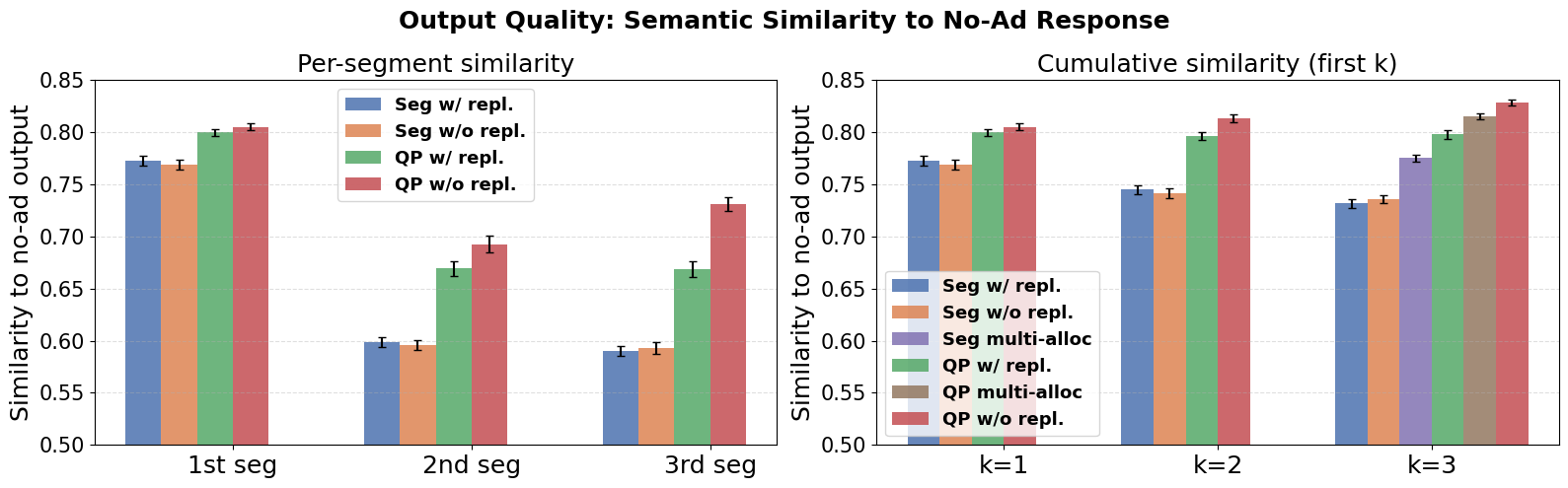}
    \caption{Output quality measured by semantic similarity to the no-ad response for Scenario 2. Left: per-segment similarity across the three generation rounds. Right: cumulative similarity over the first $k$ sentences.}
    \label{fig:output_quality}
\end{figure}

Figure \ref{fig:output_quality} shows similar results to Figure \ref{fig:output_quality_0}. For the full response case when $k=3$, QP w/o repl. performs even better than QP multi-alloc. This is because, as shown on the left side of Table \ref{tab:scenario1}, ``BookHaven'' is the only advertiser with both outstanding relevance and the highest bid. Since QP w/o repl. selects this advertiser first and is unlikely to select additional low-relevance ads, the final output typically contains only one high-relevance ad, yielding higher output quality than mechanisms that insert multiple ads. 
This indicates the need for a reserve price when most ads have low relevance.

\begin{figure}[h]
    \begin{tcolorbox}[colback=gray!10, colframe=gray!50, width=\textwidth, sharp corners]
    \scriptsize
    \begin{tabular}{p{0.3\textwidth}p{0.3\textwidth}p{0.3\textwidth}}
        \textbf{QP single-allocation (Section \ref{sec:single})}  & \textbf{QP multi-allocation (Section \ref{sec:multi})} & \textbf{Seg single-allocation (See \citep{hajiaghayi2024ad})} \\[2pt]
        \textbf{(Seg 1)} Dive into the worlds of \teal{"The Help" by Kathryn Stockett} or \teal{"I Know Why the Caged Bird Sings" by Maya Angelou}, available at \flink{BookHaven}, where you can easily find these engaging novels that, like "To Kill a Mockingbird," explore themes of race, inequality, and societal challenge, ensuring your next reading adventure is both profound and captivating.
        &
        \textbf{(Seg 1)} If you appreciated the profound themes of "To Kill a Mockingbird," you might find \teal{"The Help" by Kathryn Stockett} and \teal{"Roll of Thunder, Hear My Cry" by Mildred D. Taylor} equally compelling, both available at \flink{BookHaven} along with a vast array of literature that dives deep into social and moral complexities.
        &
        \textbf{(Seg 1)} If you appreciated the thought-provoking themes of "To Kill a Mockingbird," you might enjoy reading \teal{"The Help" by Kathryn Stockett}, a book which—much like \flink{Velora}'s seamless integration of technology into everyday life—explores complex social issues with both grace and depth.
        \\[4pt]
        \textbf{(Seg 2)} While immersing yourself in these powerful narratives, consider enhancing your reading experience with \flink{Velora}'s latest range of tablets, offering crisp displays and long battery life for uninterrupted literary indulgence.
        &
        \textbf{(Seg 2)} While exploring these narratives, consider enhancing your reading experience with \flink{Velora}'s range of elegantly designed tablets, which offer crystal-clear displays and long-lasting batteries, perfect for reading anytime, anywhere. 
        &
        \textbf{(Seg 2)} For a hassle-free way to discover more such impactful books, check out \flink{BookHaven}, your ultimate online bookstore that promises to connect every book with its perfect reader through an extensive collection, personalized recommendations, and unbeatable prices.
        \\[4pt]
        \textbf{(Seg 3)} For additional reading that captures similar themes, \teal{"The Color Purple" by Alice Walker} and \teal{"Native Son" by Richard Wright} offer further explorations into the enduring impacts of racial injustices and individual resilience.
        &
        \textbf{(Seg 3)} These books not only offer engaging stories but also mirror the profound societal observations seen in Harper Lee's work, making them perfect candidates for your next reading adventure.
        &
        \textbf{(Seg 3)} And when you need a break from reading, explore \flink{MassMart} for a \red{premium shopping experience} where you can find everything from the latest electronics to \red{bulk groceries, all at value prices, enhancing your lifestyle while meeting all your needs under one roof}.
        \\
    \end{tabular}
    \end{tcolorbox}
    \caption{Sample outputs of different mechanisms for Scenario 2. \textcolor{teal}{Teal}: useful answer information.  \textcolor{blue}{Blue}: ads link. \textcolor{red}{Red}: incoherent or forced ad insertion.}
    \label{fig:qual2}
\end{figure}

Figure \ref{fig:qual2} shows the qualitative outputs of Scenario 2, from which we can observe consistent conclusions with Scenario 1.
Specifically, QP single-allocation leverages organic content to provide useful information when no ad is selected (see Seg 3), while QP multi-allocation integrates organic and sponsored content more naturally across the full output. In contrast, Seg single-allocation is forced to insert one ad per sentence, which risks including low-relevance ads and confines useful information to the first segment that directly follows the query. This further validates the advantage of our quality-preserving mechanisms.

\subsection{Scenario 3}\label{sec:exp2}

\paragraph{Metrics Comparison.}
Scenario 3 uses the same set of advertisers as Scenario 2 but with different bids, as shown on the left side of Table \ref{tab:scenario2}. In particular, ``BookHaven'' has the highest relevance but the lowest bid, while ``MassMart'' and ``EspressoEdge'', which have low relevance, now have the highest bids. 
We use the same hyperparameter settings as Scenario 1, except that $\hat{f}(q_0)=1.5q_0^{0.8}$ for the single-allocation mechanisms.

\begin{table}[h]
    \caption{Metrics comparison results across mechanisms for Scenario 3. Left: advertiser bids and relevance scores, including organic relevance $q_0^{(1)}$. Right: auction outcomes averaged over 100 trials (reported as mean $\pm$ SE), best performance in each metric in \textbf{bold}.}
    \vspace{2mm}
    \begin{minipage}{.31\linewidth}
      \centering
      \resizebox{\textwidth}{!}{
        \begin{tabular}{ccc}
            \toprule
            Advertiser   & Bid  & $q_i^{(1)}$  \\ \midrule
            Velora       & $2$  & $0.53$ \\ 
            BookHaven    & $1$  & $0.70$ \\ 
            MassMart     & $3$  & $0.51$ \\ 
            EspressoEdge & $3$  & $0.49$ \\ 
            Organic      & $-$  & $0.79$ {\scriptsize ($\pm .002$)} \\ \bottomrule
        \end{tabular}
        }
    \end{minipage}%
    \begin{minipage}{.695\linewidth}
      \centering
        \resizebox{\textwidth}{!}{
            \begin{tabular}{lccccc}
            \toprule
            Mechanism & Revenue per Ad & Soc.\ Wel. & Relevance & KL Div. & Num.\ Ads \\ \midrule
            QP w/ repl.  & $\textbf{1.19}$ {\scriptsize ($\pm .021$)} & $4.57$ {\scriptsize ($\pm .064$)} & $1.86$ {\scriptsize ($\pm .018$)} & $0.03$ {\scriptsize ($\pm .001$)} & $2.12$ {\scriptsize ($\pm .082$)} \\ 
            QP w/o repl. & $1.17$ {\scriptsize ($\pm .024$)} & $4.06$ {\scriptsize ($\pm .019$)} & $1.82$ {\scriptsize ($\pm .020$)} & $\textbf{0.02}$ {\scriptsize ($\pm .001$)} & $1.76$ {\scriptsize ($\pm .075$)} \\ 
            QP multi-alloc     & $0.85$ {\scriptsize ($\pm .001$)} & $\textbf{4.94}$ {\scriptsize ($\pm .003$)} & $\textbf{1.91}$ {\scriptsize ($\pm .002$)} & $-$ & $2.00$ {\scriptsize ($\pm .000$)} \\ 
            Seg w/ repl.\citep{hajiaghayi2024ad}   & $1.05$ {\scriptsize ($\pm .045$)} & $4.36$ {\scriptsize ($\pm .081$)} & $1.78$ {\scriptsize ($\pm .013$)} & $0.26$ {\scriptsize ($\pm .001$)} & $3.00$ {\scriptsize ($\pm .000$)} \\ 
            Seg w/o repl.\citep{hajiaghayi2024ad} & $0.97$ {\scriptsize ($\pm .040$)} & $3.83$ {\scriptsize ($\pm .035$)} & $1.68$ {\scriptsize ($\pm .010$)} & $0.27$ {\scriptsize ($\pm .008$)} & $3.00$ {\scriptsize ($\pm .000$)} \\ 
            Seg multi-alloc\citep{hajiaghayi2024ad}   & $0.78$ {\scriptsize ($\pm .038$)} & $3.75$ {\scriptsize ($\pm .033$)} & $1.62$ {\scriptsize ($\pm .010$)} & $-$ & $3.00$ {\scriptsize ($\pm .000$)} \\ 
            \bottomrule
            \end{tabular}
        }
    \end{minipage}
    \label{tab:scenario2}
\end{table}

The right side of Table \ref{tab:scenario2} shows consistent results with previous scenarios, indicating that even when high-relevance ads are filtered out due to low bids, the introduction of organic content serves as a complement to guarantee higher relevance and social welfare, while the reserve price guarantees a revenue lower bound, demonstrating the robustness of our method.

\paragraph{Quality Comparison.} We measure output quality in Figure \ref{fig:output_quality_2}.

\begin{figure}[h]
    \centering
    \includegraphics[width=\textwidth]{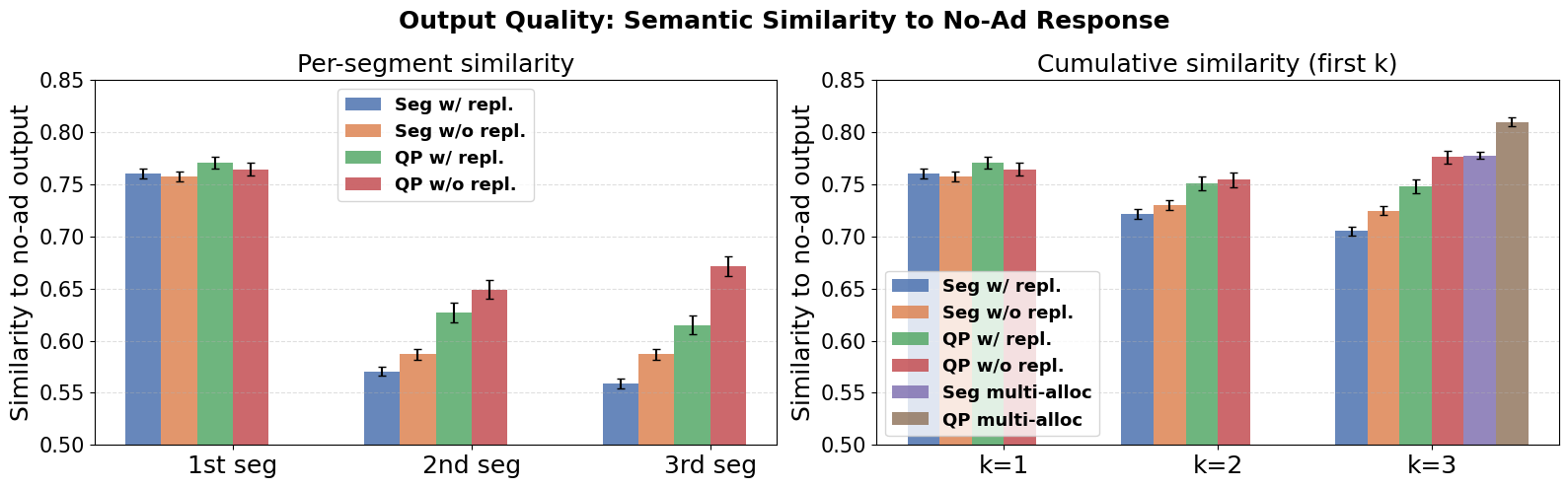}
    \caption{Output quality measured by semantic similarity to the no-ad response for Scenario 3. Left: per-segment similarity across the three generation rounds. Right: cumulative similarity over the first $k$ sentences.}
    \label{fig:output_quality_2}
\end{figure}

As shown in Figure \ref{fig:output_quality_2}, relevance is similar across mechanisms in the first sentence. This is because low-relevance ads are now likely to be selected due to their high bids, and since the first sentence directly follows the query, it tends to directly answer the query and include useful information regardless of the mechanism. However, for the following sentences, our QP methods show a substantial increase in relevance due to the use of organic content. This further verifies that when ads have low relevance, incorporating organic content is crucial for improving output quality.
For the full response case when $k=3$, the multi-allocation mechanism generally outperforms single-allocation. This is because most selected ads have low relevance, and multi-allocation can better combine them into a more coherent output. We further illustrate this in Figure \ref{fig:qual3}.

\begin{figure}[h]
    \begin{tcolorbox}[colback=gray!10, colframe=gray!50, width=\textwidth, sharp corners]
    \scriptsize
    \begin{tabular}{p{0.3\textwidth}p{0.3\textwidth}p{0.3\textwidth}}
        \textbf{QP single-allocation (Section \ref{sec:single})}  & \textbf{QP multi-allocation (Section \ref{sec:multi})} & \textbf{Seg multi-allocation (See \citep{hajiaghayi2024ad})} \\[2pt]
        \textbf{(Seg 1)} As you immerse yourself in books like \teal{"The Help" by Kathryn Stockett} or \teal{"Roll of Thunder, Hear My Cry" by Mildred D. Taylor}, which echo themes of justice and moral growth found in "To Kill a Mockingbird," consider enhancing your reading experience with a soothing cup from \flink{EspressoEdge}, where every beverage brings a touch of luxury to your literary journey. 
        &
        \textbf{(Seg 1)} If you enjoyed the deep exploration of social issues in "To Kill a Mockingbird," you might find \teal{"The Help" by Kathryn Stockett} equally compelling, as it delves into the complexities of race and class in 1960s Mississippi. 
        &
        \textbf{(Seg 1)} For those who appreciated the impactful narrative and moral challenges in "To Kill a Mockingbird," you might enjoy reading \teal{"The Help" by Kathryn Stockett}, which offers nuanced dialogues about race and class. 
        \\[4pt]
        \textbf{(Seg 2)} These recommendations, alongside explorations like \teal{"A Time to Kill" by John Grisham}, \teal{"The Secret Life of Bees" by Sue Monk Kidd}, and \teal{"Go Set a Watchman" by Harper Lee}, can further enrich your understanding of civil rights and moral challenges through powerful storytelling.
        &
        \textbf{(Seg 2)} After picking up your copy from \flink{MassMart}, where deals on literature allow you to save big while supporting community values, you could discuss the book’s themes with friends over a rich espresso at \flink{EspressoEdge}. 
        &
        \textbf{(Seg 2)} While exploring these profound themes, why not set the perfect ambiance with a creamy latte from \flink{EspressoEdge}, enhancing your reading experience with a touch of luxury?
        \\[4pt]
        \textbf{(Seg 3)} Whether shopping for books or \red{home essentials, consider visiting \flink{MassMart}}, where every purchase promises quality and value in the heart of your community.
        &
        \textbf{(Seg 3)} This setting, much like the flavorful layers of their premium Arabica blends, provides the perfect backdrop for unraveling the rich narratives and moral questions these books present.
        &
        \textbf{(Seg 3)} If you're looking for a place to purchase these books, remember that \flink{MassMart} offers an array of quality books at great prices, and while there, explore the latest tech from \flink{Velora} to enjoy an effortlessly smooth digital reading experience on a sleek new tablet.
        \\
    \end{tabular}
    \end{tcolorbox}
    \caption{Sample outputs of different mechanisms for Scenario 3. \textcolor{teal}{Teal}: useful answer information.  \textcolor{blue}{Blue}: ads link. \textcolor{red}{Red}: incoherent or forced ad insertion.}
    \label{fig:qual3}
\end{figure}

Figure \ref{fig:qual3} shows that in this scenario, although QP single-allocation benefits from organic content to include more useful information (see Seg 2), the low-relevance ads in Seg 3 appear as incoherent, forced insertions. In contrast, multi-allocation integrates low-relevance ads more naturally into the surrounding context and can place two ads in a single sentence (see Seg 2 for QP multi-allocation and Seg 3 for Seg multi-allocation).

\subsection{Scenario 4}\label{sec:exp3}

\paragraph{Metrics Comparison.}
Scenario 4 uses a larger number of ads with uniform bids, as shown in Table \ref{tab:alloc_11} with their corresponding initial relevance $q_i^{(1)}$. The organic document still has the highest relevance among all 11 ads, justifying its use as the screening threshold for the eligible ads set. We use the same hyperparameter settings as previous scenarios, except that $\hat{f}(q_0)=0.7q_0^{0.8}$ for single-allocation and $\hat{f}_{\text{multi}}(q_0)=0.9\hat{f}(q_0)$ for multi-allocation, since the bids are lower in this scenario.

\vspace{-0.3cm}

\begin{table}[h]
    \caption{Advertiser bids and initial relevance scores $q_i^{(1)}$ for Scenario 4. Organic $q_0^{(1)}$ reported as mean $\pm$ SE over 100 trials.}
    \vspace{2mm}
    \label{tab:alloc_11}
    \centering\small
    \begin{tabular*}{0.9\linewidth}{@{\extracolsep{\fill}}lcccccc}
        \toprule
        Adv & Velora & BookHaven & MassMart & EspressoEdge & SocialHub & ColaBubbles \\
        \midrule
        Bids & $1$ & $1$ & $1$ & $1$ & $1$ & $1$ \\
        $q_i^{(1)}$ & $0.53$ & $0.70$ & $0.51$ & $0.49$ & $0.48$ & $0.52$ \\
        \bottomrule
    \end{tabular*}
    \vspace{2mm}
    \begin{tabular*}{0.9\linewidth}{@{\extracolsep{\fill}}lccccc||c}
        \toprule
        Adv & FizzyPop & SkyTech & AeroDynamics & MusicStream & BrainChips & Organic \\
        \midrule
        Bids & $1$ & $1$ & $1$ & $1$ & $1$ & $-$ \\
        $q_i^{(1)}$ & $0.53$ & $0.50$ & $0.52$ & $0.52$ & $0.52$ & $0.78$ {\scriptsize ($\pm .003$)} \\
        \bottomrule
    \end{tabular*}
\end{table}
\vspace{-0.7cm}
 
\begin{table}[h]
    \caption{Auction outcomes for Scenario 4  over 100 trials (reported as mean ± SE). The best performance in each metric is highlighted in \textbf{bold}.}
    \vspace{2mm}
    \centering
    \resizebox{\textwidth}{!}{%
    \begin{tabular}{lccccc}
        \toprule
        Mechanism & Revenue per Ad & Soc.\ Wel. & Relevance & KL Div. & Num.\ Ads \\ \midrule
        QP w/ repl.  & $\textbf{0.5164}$ {\scriptsize ($\pm .0153$)} & $1.8750$ {\scriptsize ($\pm .0135$)} & ${2.0948}$ {\scriptsize ($\pm .0156$)} & ${0.0018}$ {\scriptsize ($\pm .0001$)} & $1.91$ {\scriptsize ($\pm .0950$)} \\ 
        QP w/o repl. & $0.5040$ {\scriptsize ($\pm .0169$)} & $1.8075$ {\scriptsize ($\pm .0057$)} & $2.0650$ {\scriptsize ($\pm .0189$)} & $0.0011$ {\scriptsize ($\pm .0000$)} & $1.74$ {\scriptsize ($\pm .1006$)} \\ 
        QP multi-alloc     & $0.2317$ {\scriptsize ($\pm .0440$)} & $\textbf{3.6580}$ {\scriptsize ($\pm .0634$)} & $\textbf{2.6886}$ {\scriptsize ($\pm .0738$)} & $-$ & $2.79$ {\scriptsize ($\pm .2183$)} \\ 
        Seg w/ repl.\citep{hajiaghayi2024ad}   & $0.4747$ {\scriptsize ($\pm .0184$)} & $1.7207$ {\scriptsize ($\pm .0127$)} & $1.7207$ {\scriptsize ($\pm .0127$)} & $\textbf{0.0000}$ {\scriptsize ($\pm .0000$)} & $3.00$ {\scriptsize ($\pm .0000$)}\\ 
        Seg w/o repl.\citep{hajiaghayi2024ad} & $0.4574$ {\scriptsize ($\pm .0174$)} & $1.6886$ {\scriptsize ($\pm .0088$)} & $1.6886$ {\scriptsize ($\pm .0088$)} & $\textbf{0.0000}$ {\scriptsize ($\pm .0000$)} & $3.00$ {\scriptsize ($\pm .0000$)} \\ 
        Seg multi-alloc\citep{hajiaghayi2024ad}   & $0.4741$ {\scriptsize ($\pm .0152$)} & $1.5920$ {\scriptsize ($\pm .0092$)} & $1.5920$ {\scriptsize ($\pm .0092$)} & $-$ & $3.00$ {\scriptsize ($\pm .0000$)}\\ 
        \bottomrule
    \end{tabular}}
    \label{tab:scenario3}
\end{table}
\vspace{-0.1cm}

Table \ref{tab:scenario3} shows conclusions similar to previous scenarios. The KL divergence of Seg single-allocation is zero because the conditional distribution at segment $t$ satisfies
$
\P_{\text{Seg-single}}(y^{(t)}\mid x, y^{(1:t-1)})=\sum_{i=1}^n \hat{q}_i^{(t)}\,\P_{\text{ref}}(y^{(t)}\mid x, y^{(1:t-1)}, a_i),
$
where $\hat{q}_i^{(t)} \propto b_i^{(t)} q_i^{(t)}$. When bids are uniform, $\hat{q}_i^{(t)} = q_i^{(t)}$, and thus the induced distribution equals the organic one. Our QP methods also maintain negligible KL divergences through the KL penalization in the social welfare objective. While QP multi-alloc yields lower revenue per ad, as the VCG mechanism prioritizes social welfare rather than revenue, its output quality is much higher than the other mechanisms, similar to Scenario 3, as shown in the following.

\paragraph{Quality Comparison.} We measure output quality in Figure \ref{fig:output_quality_3}.

\begin{figure}[h]
    \centering
    \includegraphics[width=\textwidth]{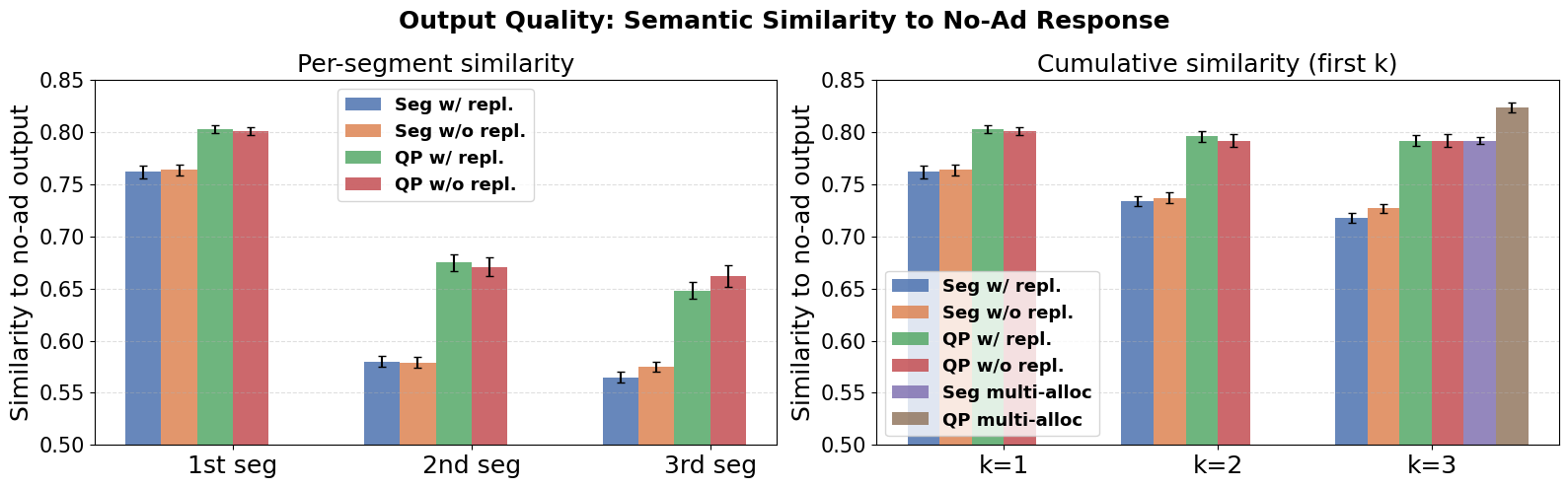}
    \caption{Output quality measured by semantic similarity to the no-ad response for Scenario 4. Left: per-segment similarity across the three generation rounds. Right: cumulative similarity over the first $k$ sentences.}
    \label{fig:output_quality_3}
\end{figure}
\vspace{-0.3cm}
Figure \ref{fig:output_quality_3} shows a similar pattern to previous scenarios. Sampling with and without replacement yields nearly identical results here, since repeated ads are rarely drawn from a large candidate pool.
\end{document}